\documentclass[11pt]{article}

\usepackage{setspace}

    \usepackage{amsmath, amssymb,amsfonts,xspace,graphicx,relsize,bm,mathtools,xcolor}
    \usepackage{mathrsfs}
    \usepackage[pagebackref]{hyperref}
    \usepackage{enumerate}
    \usepackage{bm}
  \usepackage{amsthm,cleveref}
    \usepackage{lipsum} 
    \newtheorem{theorem}{Theorem}
    \newtheorem{result}{Result}
    \newtheorem{remark}{Remark}

\usepackage{url}
\def\01{\{0,1\}}

\newcommand{\eps}{\varepsilon}
\newcommand{\ket}[1]{|#1\rangle}
\newcommand{\bra}[1]{\langle#1|}
\newcommand{\ketbra}[2]{|#1\rangle\langle#2|}

\newcommand{\braket}[2]{\langle#1|#2\rangle}

\newcommand{\Id}{I}
\newcommand{\E}{\mathbb{E}}
\newcommand{\Cc}{{\mathcal C}} 
\newcommand{\Exp}{{\mathbb{E}}}

\newcommand{\X}{\ensuremath{\mathcal{X}}}

\newcommand{\D}{\ensuremath{\mathcal{D}}}

\newcommand{\A}{\ensuremath{\mathcal{A}}}

\newcommand{\id}{\ensuremath{\mathbb{I}}}

\usepackage{enumitem}
\DeclareMathOperator{\poly}{poly}
\newcommand{\Z}{\mathcal{Z}}

\DeclareMathOperator{\polylog}{polylog}

\newtheorem{definition}[theorem]{Definition}

\newtheorem{fact}[theorem]{Fact}
\newtheorem{lemma}[theorem]{Lemma}

\newtheorem{proposition}[theorem]{Proposition}

\newcommand{\pmset}[1]{\{-1,1\}^{#1}} 

\usepackage{tikz}
\usetikzlibrary{positioning}

\newcommand{\EM}{\mathsf{EM}}
\newcommand{\Var}{\mathsf{Var}}
\newcommand{\sign}{\mathsf{sign}}

\usepackage{multicol}

\usepackage[margin=1in]{geometry}
\hypersetup{
	colorlinks,
	linkcolor={red!100!black},
	citecolor={blue!100!black},
}
\usepackage{tcolorbox}
\usepackage{tocbasic}

\usepackage{footnote}

\newcommand{\QSQ}{\mathsf{QSQ}}
\newcommand{\PAC}{\mathsf{PAC}}
\newcommand{\QPAC}{\mathsf{QPAC}}
\newcommand{\Qstat}{\mathsf{Qstat}}
\newcommand{\Stat}{\mathsf{Stat}}
\newcommand{\Tr}{\mathsf{tr}}
\newcommand{\QSD}{\mathsf{QQC}}
\newcommand{\TVD}{d_{\text{TV}}}
\newcommand{\QSDA}{\mathsf{QSD}}
\newcommand{\QAC}{\mathsf{QAC}}
\newcommand{\SQ}{\mathsf{SQ}}
\newcommand{\SWAP}{\mathrm{\mathsf{SWAP}}}
\newcommand{\TRD}{d_{\Tr}}
\newcommand{\HD}{d_{\mathsf{H}}}
\newcommand{\QSSD}{\mathsf{QSSD}}

\newcommand{\ML}{\ensuremath{\mathsf{ML}}}

\newcommand{\C}{\ensuremath{\mathbb{C}}}

\usepackage{algpseudocode}

\begin{document}

\title{On the Role of Entanglement and Statistics in Learning}

\author{Srinivasan Arunachalam\thanks{IBM Quantum, Almaden Research Center, \textsf{Srinivasan.Arunachalam@ibm.com}} \and
Vojt{\v e}ch Havl{\'i}{\v c}ek\thanks{IBM Quantum, T.J. Watson Research Center, \textsf{
Vojtech.Havlicek@ibm.com}}
 \quad \and
Louis Schatzki\thanks{Electrical and Computer Engineering, University of Illinois Urbana-Champaign, \textsf{louisms2@illinois.edu}}\vspace{1mm}
 \and
}

\maketitle

\begin{abstract}
We make progress in understanding the relationship between learning models with access to entangled, separable and statistical measurements in the quantum statistical query ($\QSQ$) model. We show the following~results.

\vspace{3mm}

\textbf{Entangled versus separable measurements.} The goal is to learn an unknown~$f$ from the concept class $\Cc\subseteq \{f:\01^n\rightarrow [k]\}$ given copies of  $\frac{1}{\sqrt{2^n}}\sum_x \ket{x,f(x)}$.  We show that, if $T$ copies suffice to learn~$f$ using entangled measurements, $O(nT^2)$ copies suffice to learn~$f$ using only separable measurements.
\vspace{3mm}

\textbf{Entangled versus statistical measurements} The goal is to learn a function $f \in \Cc$ given access to separable measurements or statistical measurements.   We exhibit a concept class $\Cc$ based of degree-$2$ functions with exponential separation between $\QSQ$ learning and quantum learning with entangled measurements (even in the presence of noise). This proves the ``quantum analogue" of the seminal result of Blum et al.~\cite{blum2003noise} that separates classical $\SQ$ learning from classical $\PAC$ learning with classification~noise. 
\vspace{3mm}

\textbf{$\QSQ$ lower bounds for learning states.} We introduce a quantum statistical query dimension ($\QSDA$), and use it to give lower bounds on the $\QSQ$ complexity of learning. We prove superpolynomial $\QSQ$ lower bounds for testing purity of quantum states, shadow tomography, learning coset states for the Abelian hidden subgroup problem, degree-$2$ functions, planted bi-clique states and learning output states of Clifford circuits of depth $\polylog(n)$. We also show that an extension of $\QSDA$ \textit{characterizes} the complexity of general search problems.
\vspace{3mm}

\textbf{Further applications.} We give \emph{unconditional} separation between weak and strong error mitigation and prove lower bounds for learning distributions in the $\QSQ$ model. Prior works by Quek et al.~\cite{quek2022exponentially}, Hinsche et al.~\cite{hinsche2022single} and Nietner et al.~\cite{sweke23} proved analogous results \emph{assuming} diagonal measurements and our work removes this~assumption.

\end{abstract}

\newpage

\newpage

\begin{spacing}{0.15}
\tableofcontents
\end{spacing}

\newpage

\section{Introduction}
 Machine learning ($\ML$) has emerged as one of successful parts of artificial intelligence  with wide-ranging applications in computer vision, image recognition, natural  language processing. More recently, $\ML $ has been used  in popular applications  such as AlphaGo and Alpha zero (to play the games of Go and chess), chatGPT (to mimic a human conversation) and Alphafold (for solving some hard instances of protein folding). Simultaneously, understanding the power of quantum physics for $\ML$ has received much attention in the quantum computing community.  Many quantum algorithms have been proposed for practically relevant $\ML$ tasks such as clustering, recommendation systems,  linear algebra, convex optimization, support-vector machines, kernel-based methods, topological data analysis~\cite{lloyd:clustering,chia:dequantizeall,brandao2017quantum,rebentrost2014quantum,kerenidis2016quantum,havlivcek2019supervised,mcardle2022streamlined,gyurik2022towards,van2020convex,tang:recommendation}. There are several surveys dedicated to understanding the power of quantum methods for $\ML$~\cite{biamonte2017quantum,cerezo2022challenges,arunachalam2017guest,schuld2015introduction,anshu23}.

Quantum learning theory provides a theoretical framework to understand quantum advantages in $\ML$. Here, there is a concept class $\Cc$ which is a collection of $n$-qubit quantum states, a learner is provided with several copies of $\rho \in \Cc$, performs an arbitrary entangled operation on $\rho^{\otimes T}$ and the goal is to learn $\rho$ well-enough.  This framework encompasses several results in quantum learning  such as tomography, shadow tomography, learning interesting classes of states, learning an unknown distribution and functions encoded as a quantum state~\cite{o2016efficient,haah2016sample,aaronson:shadow,huangpreskill,montanaro2017learning,arunachalam2022optimal,HTFS22,DBLP:conf/focs/ChenCH021,hinsche2022single,quek2022exponentially,sweke23,arunachalam2017guest,anshu23}. 

A natural concern when considering the near-term implementation of such quantum learning algorithms is, it is infeasible to prepare several copies of $\rho$ and furthermore, perform arbitrary entangled measurements $\rho$. More recently, motivated by near-term implementations,~\cite{arunachalam2020quantum} introduced the model of \emph{quantum statistical query} ($\QSQ$) learning, to understand the power of measurement statistics for learning, with variations of it finding applications in~\cite{quek2022exponentially,hinsche2022single,sweke23,du2021learnability,gollakota2022hardness,angrisani2022quantum}. In the $\QSQ$ model, suppose a learning algorithm wants to learn an unknown $n$-qubit quantum state $\rho$, a learning algorithm can perform $\poly(n)$-many \emph{efficiently}-implementable two-outcome measurements $\{M_i,\id-M_i\}$ with noise and the goal is to learn the unknown $\rho$ well enough. Clearly this model is weaker than the model when given access to $\rho^{\otimes T}$, since the learner is only allowed access to \emph{expectation values} over a single copy of $\rho$.

In this work, we primarily consider concept classes constructed from Boolean functions. In Valiant's probabily approximately correct (PAC) learning framework,  a \emph{concept class} $\Cc\subseteq \{c:\01^n\rightarrow \01\}$ is a collection of Boolean functions. In the  PAC  model,\footnote{For simplicity, we discuss PAC learning under the \emph{uniform}-distribution, i.e., $x$ is drawn uniformly from $\01^n$.} a learning algorithm is given many uniformly random $(x^i,c^\star(x^i))$ where $c^\star\in \Cc$ is unknown and it uses these to learn $c^\star$ approximately well. Bshouty and Jackson introduced the \emph{quantum $\PAC$} ($\QPAC$) model~\cite{DBLP:conf/colt/BshoutyJ95} wherein a quantum learner is given \emph{quantum examples} $\ket{\psi_{c^\star}}^{\otimes T}$, i.e., coherent superpositions
$
\ket{\psi_{c^\star}}=\frac{1}{\sqrt{2^n}}\sum_{x}\ket{x,c^\star(x)},
$
 and it needs to learn the unknown $c^\star\in\Cc$ well enough. The complexity measure here is the \emph{sample complexity}, i.e., copies of classical or quantum examples used by the algorithm. There have been works that have looked at this model and proven positive and negative results for learning function classes (see~\cite{arunachalam2017guest} for a survey).  Surprisingly, in~\cite{arunachalam2020quantum} they observed that, many positive results using quantum examples can be transformed into algorithms in the weaker $\QSQ$ framework. This motivates the following two open questions:
\begin{quote}
\hspace{15mm}\textbf{\emph{1.}} \emph{Are entangled measurements needed for learning function classes?}

\hspace{15mm}\textbf{\emph{2.}} \emph{Do measurement statistics suffice for learning function classes?}   
\end{quote}

\subsection{Main results}
In this work, we resolve both of these questions. 
We show that $(i)$ for learning \emph{Boolean} function classes the sample complexity of learning with entangled measurements and separable measurements are polynomially related and $(ii)$ there is an exponential separation between  learning with separable measurements (even in the presence of classification noise) and learning with just measurement statistics. We now discuss these results in more detail.

\paragraph{Entangled versus Separable measurements.} Understanding the role of entangled measurements in quantum information has received attention recently. Bubeck et al.~\cite{DBLP:conf/focs/BubeckC020} gave a property testing task for which entangled measurements are \emph{necessary} for obtaining the optimal bounds. More recently, for learning classes of arbitrary \emph{quantum states} (i.e., not necessarily states constructed from function classes), there were two recent works by~\cite{DBLP:conf/focs/ChenCH021,huang2022quantum} which showed  exponential separation for learning properties of quantum states when given access entangled measurements in comparison to separable measurements. Here, we study if similar separations exist when considering \emph{function} classes, a small subset of all quantum states.  

Our first result shows that in order to exactly learn a function class, every learning algorithm using entangled measurements  can be transformed into a learning algorithm using just separable measurements with a polynomial overhead in sample complexity. In contrast, if the goal is to learn a \emph{property} about the unknown function, then entangled measurements can reduce the sample complexity exponentially compared to separable measurements.
\begin{result}
\label{res1}
    For a concept class $\Cc\subseteq \{c:\01^n\rightarrow \01\}$, if $T$ copies of $\ket{\psi_c}$ suffice to learn $c$, then $O(nT^2)$ copies to learn $c$ using only separable measurements.
\end{result}

\paragraph{$\QSQ$ versus noisy-$\QPAC$ learning.} In~\cite{arunachalam2020quantum} they ask if there is a natural class of Boolean functions for which, $\QSQ$ learning can be separated from $\QPAC$ learning. Classically it is well-known that parities separates $\SQ$ learning from $\PAC$ learning. In~\cite{arunachalam2020quantum}, it was observed that the class of parities, juntas, DNF formulas are learnable in the $\QSQ$ framework and a candidate class to separate $\QSQ$ from quantum-$\PAC$ was unclear. Furthermore, Kearns posed the question if $\SQ$ learning is equal to $\PAC$ learning with classification noise.  The seminal result of Blum et al.~\cite{blum2003noise} resolves this question by showing that the class of parity functions on $O((\log n)\cdot \log \log n)$ separates these two models of learning (under constant noise rate).  This motivates the following questions:
\begin{enumerate}[label={$(\alph*)$}]
    \item  In the noisy-quantum $\PAC$ model~\cite{DBLP:conf/colt/BshoutyJ95,arunachalam2018optimal}, a learning algorithm is given copies of 
\begin{align}
\label{eq:noisyintroexample}
    \ket{\psi^n_{c^\star}}=\frac{1}{\sqrt{2^n}}\sum_{x\in \01^n}\ket{x} \big(\sqrt{1-\eta}\ket{c^\star(x)}+\sqrt{\eta}\ket{\overline{c^\star(x)}}\big).
    \end{align}
    the goal is to learn $c^\star$. Is there a class that separates  noisy-quantum $\PAC$ from $\QSQ$ learning?
    \item Admittedly, the class constructed by Blum et al.~\cite{blum2003noise} is ``unnatural", can we obtain the   separation in $(a)$ for a \emph{natural} concept class?
    \item Does such a separation hold for non-constant error rate $\eta$?
\end{enumerate}
Here, we describe a natural problem that witnesses this separation and resolves the questions above. 
\begin{result}
\label{res2}
    There is a concept class of $n$-bit Boolean functions, based out of degree-$2$ functions that can be solved in using quantum examples even in the presence of $\eta$-classification noise in time $\poly(n,1/(1-2\eta))$, whereas every $\QSQ$ algorithm requires $2^{\Omega(n)}$ queries to learn $\Cc$. 
\end{result}

\subsection{Further applications}
Using our $\QSQ$ lower bounds for learning quadratic functions, we present two applications. First we give an exponential separation between weak and strong error mitigation, resolving an open question of Quek et al.~\cite{quek2022exponentially} who proved the same separation assuming the $\QSQ$ observables are diagonal. Second, we show super-polynomial lower bounds for learning output distributions (in the computational basis) of $n$-qubit Clifford circuits of depth $\omega(\log n)$ and Haar random circuit of depth-$O(n)$. This extends the work of~\cite{hinsche2022single,sweke23} who proved these lower bounds for $\QSQ$ algorithms wherein the observables are diagonal.

\emph{Error mitigation.} Error mitigation ($\EM$) was introduced as an algorithmic technique to reduce the noise-induced in near-term quantum devices, hopefully with a small overhead, in comparison to building a full-scale fault-tolerant quantum computer~\cite{temme2017error}. In recent times, $\EM$ has obtained a lot of attention with several works understanding how to  obtain \emph{near-term} quantum speedups as a surrogate to performing error correction.  More formally, an $\EM$ algorithm $\A$ takes as input a quantum circuit $C$, noise channel $\mathcal{N}$ and copies of $\ket{\psi'}=\mathcal{N}(C)\ket{0^n}$. In a strong $\EM$ protocol, $\A$ needs to produce samples from a distribution $D$ that satisfies $\TVD(D,\{\langle x| C|0^n\rangle^2\}_x)\leq \varepsilon$ and in the   weak $\EM$ setting, given observables $M_1,\ldots,M_k$ the goal is to approximate $\langle \psi|M_i|\psi\rangle$ upto $\varepsilon$-error. In~\cite{quek2022exponentially}, they asked the question: how large should $k$ be in order to simulate weak $\EM$ by strong $\EM$? They  show that \emph{when $M_i$s are diagonal}, then $k= 2^{\Omega(n)}$, i.e., they gave an exponential separation between weak and strong $\EM$. In this work, our main contribution is to use Result~\ref{res2} to remove the assumption and show an exponential separation unconditionally between weak and strong $\EM$.

\emph{Learning distributions.} Recently, the works of Hinsche et al.~\cite{hinsche2022single} and Nietner et al.~\cite{sweke23} initiated the study of learning output distributions of quantum circuits. In particular, they considered the following \emph{general} question: Let $\ket{\psi_U}=U\ket{0^n}$ where $U\in \mathcal{U}$ and $\mathcal{U}$ is a family of interesting unitaries and let $P_U(x)=|\langle x|U|0^n\rangle|^2$. How many $\QSQ$ queries does one need to learn the $P_U$ to total variational distance at most $\varepsilon$? To this end, the works of~\cite{hinsche2022single,sweke23} looked at \emph{diagonal} $M$s, i.e., $M=\sum_X \phi(x) \ketbra{x}{x}$ for $\phi:\01^n\rightarrow [-1,1]$ and showed the hardness of approximately learning $P_U$ for $\mathcal{U}$ being $\omega(\log n)$-depth Clifford circuits and  depth-$d\in \{\omega(\log n), O(n)\}$ and  $d\rightarrow \infty$-depth Haar random circuits.\footnote{They in fact prove that learning even a $(1-\exp(-n))$-{frac}tion of the circuits in these circuit families is hard in $\QSQ$ model when restricted to diagonal observables.} In this work, we improve upon their lower bounds by removing the assumption that $M$ is diagonal and prove a \emph{general} $\QSQ$ lower bounds for these circuit families that is considered in their work. We also observe that learning the output states of \emph{constant}-depth circuit can be done in polynomial-time using $\QSQ$ queries.

\subsection{Proof overview}
In this section we give a brief overview of the results we described above. 

\paragraph{Relating entangled and separable learning.} Our starting point towards proving this result is, that one could use a result of Sen~\cite{sen2006random} that, given copies of $\ket{\psi_{c^\star}}$, one could apply random measurements on single copies of this state and produce an $h$ that is approximately close to $c^\star$ using at most $T=(\log |C|)/\varepsilon$ copies of $\ket{\psi_{c^\star}}$.\footnote{This idea was used in an earlier work of Chung and Lin~\cite{chung2018sample} as well, but they weren't concerned with entangled and separable measurements.} So for separable learning, by picking $\varepsilon=\eta_{min}$ as the minimum distance between concepts in $\Cc$, one could exactly learn $\Cc$ using $T$ quantum examples. Proving a lower bound on entangled learning $\Cc$ is fairly straightforward as well: first observe that $(\log |\Cc|/n)$ is a lower bound on learning (since each quantum example gives $n$ bits if information and for exact learning one needs $\Omega(\log |\Cc|)$ bits of information) and also observe that $1/\eta_{m}$ is a lower since to distinguish just between $c,c'\in \Cc$ that satisfy $\Pr_x[c(x)=c'(x)]=1-\eta_{m}$, one needs $1/\eta_{m}$ copies of the unknown state. Putting this separable upper bound and entangled lower bound together gives us $\textsf{SepExact}(\Cc)\leq n\cdot \textsf{EntExact}(\Cc)^2$ for all $\Cc$. We further improve the entangled lower bound as follows: let $\eta_a=\Exp_{c,c'\in \Cc}\Pr_x [c(x)\neq c'x)]$, then using a information-theoretic argument (inspired by a prior work~\cite{arunachalam2018optimal}) one can show that the entangled sample complexity of exact learning is at least $\max \{1/\eta_m, (\log |\Cc|)/(n\eta_a)\}$. Putting this entangled lower bound with the separable upper bound, we get that
$$
\textsf{SepExact}\leq  O\Big(n\cdot \textsf{EntExact}\cdot \min\big\{\eta_{\textsf{a}}/\eta_{\textsf{m}}\ , \ \textsf{EntExact}\Big\}\Big).
$$
It is not hard to see that this bound is optimal as well for 
the class of degree-$2$ functions, i.e., 
\begin{align}
\label{eq:deg2functionclass}
\Cc=\{f(x)=x^\top A x \pmod 2: A\in \mathbb{F}_2^{n\times n}\}.
\end{align}
For this class $\eta_a=\eta_m=O(1)$ and recently it was shown~\cite{arunachalam2022optimal} that  $\textsf{SepExact}=\Theta(n^2)$ and $\textsf{EntExact}=\Theta(n)$. 

\paragraph{A combinatorial parameter to lower bound $\QSQ$ complexity.} A fundamental issue in proving our $\QSQ$ result is, what techniques could one use to prove these lower bounds? Prior to our work, in~\cite{arunachalam2020quantum} they introduced two new techniques based on differential privacy and communication complexity that give lower bounds on $\QSQ$ complexity. However, both these lower bounds are exponentially weak! In particular, the lower bounds that they could prove were linear in $n$ for learning an $n$-bit concept class.  Classically, there have been a sequence of works~\cite{Feldman13,Feldman16,feldman2017statistical} with the goal of proving $\SQ$ lower bounds and finally the notion of \emph{statistical dimension} was used to obtain close-to-optimal bounds for $\SQ$ learning certain concept classes and the breakthrough works of~\cite{feldman2017statistical} used it to settle the complexity of learning the planted $k$-biclique~distribution. 

 In this work, our technical contribution is a combinatorial parameter to lower bound $\QSQ$ complexity akin to the classical parameter. To this end, we follow a three-step approach. 
\begin{enumerate}
    \item We show that an algorithm $\A$ that learns a concept class below error $\varepsilon$ in trace distance using $\Qstat$ queries of tolerance $\tau$ can also be used to solve the following decision problem: for a fixed $\sigma$ such that $\min_{\rho \in \Cc}d_{\Tr}(\rho,\sigma) > 2(\tau+\eps)$, decide if an unknown state is either some $\rho \in \Cc$ or equals $\sigma$. Calling $\QSD$ the complexity of such decision problem, we show that:\footnote{A similar argument appeared in~\cite{sweke23} for diagonal-$\QSQ$ complexity. We want to thank the authors for discussing their work with us during the completion of our work.}
    $$
    \QSQ(\Cc)\geq \max_\sigma \{\QSD(\Cc, \sigma) -1: \min_{\rho \in \Cc}d_{\Tr}(\rho,\sigma) > 2(\tau+\eps)\}.
    $$
    \item Next, we define the notion of \emph{quantum statistical dimension} $\QSDA$: for $\tau >0$, a class of states~$\Cc$ and a $\sigma \notin \Cc$,
the $\QSDA_\tau(\Cc, \sigma)$ is the smallest integer such that there exists a distribution $\nu$ over $\Qstat$ queries $M$ satisfying $\Pr_{M \sim \nu} \left[|\Tr(M(\rho - \sigma))| > \tau \right] \geq 1/d$ for all $\rho \in \Cc$.
From an operational perspective $\QSDA$ is natural, as it can be viewed as the smallest expected number of observables that can distinguish all states in $\Cc$ from $\sigma$. We then show that if the decision algorithm succeeds with probability at least $1-\delta$, we have that: 
$$
\QSD(\Cc, \sigma)\geq (1-2\delta)\QSDA_\tau(\Cc, \sigma).
$$ 
\item  Even with this lower bound, proving bounds on $\QSDA(\Cc, \sigma)$ is non-trivial. To this end, we further give two lower bounding techniques for $\QSDA(\Cc, \sigma)$, one based on the variance of $\Qstat$ queries across $\Cc$ (inspired by the work of Kearns~\cite{kearns1998efficient}) and one based on average correlation (inspired by the work of Feldman~\cite{Feldman16}). We define two combinatorial quantities $\Var(\Cc)$ and $\QAC(\Cc, \sigma)$ which can be associated with every class and use it to lower bound $\QSDA(\Cc)$.

\end{enumerate}
Putting the three points together, the $\QSQ$ complexity of learning can be lower bounded by the variance bound and the average correlation bound as summarized in the figure below.

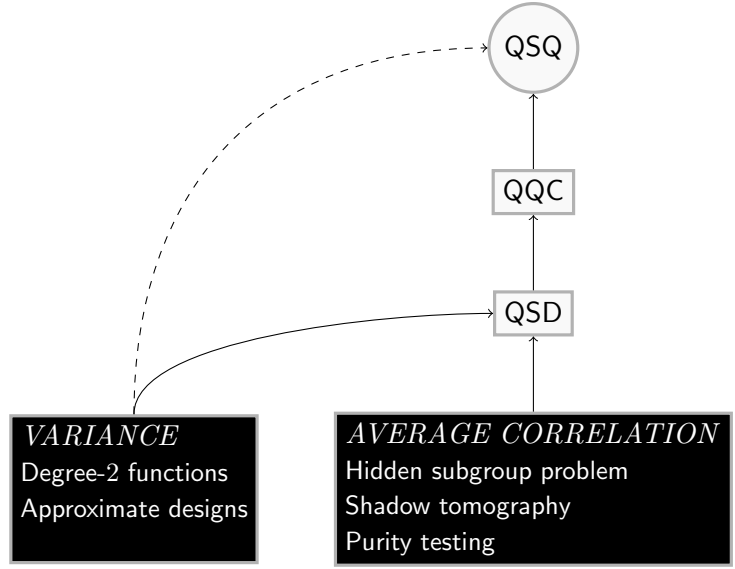
\begin{figure}[t]
\centering
\begin{tikzpicture}[
roundnode/.style={circle, draw=gray!60, fill=gray!5, very thick, minimum size=7mm},
squarednode/.style={rectangle, draw=gray!60, fill=gray!5, very thick, minimum size=5mm},
rootnode/.style={rectangle,  align=left, draw=gray!60, fill=black!100, text=white, very thick, minimum size=5mm},
]

\node[squarednode]      (qsd)                             {$\QSDA$};
\node[squarednode]      (qqc) [above=of qsd]                         {$\QSD$};

\node[roundnode]        (qsq)       [above=of qqc] {$\QSQ$};

\node[rootnode]      (avc) [below=of qsd]                         {\emph{AVERAGE CORRELATION} \\ \small \textsf{Hidden subgroup problem} \\ \small \textsf{Shadow tomography} \\ \small \textsf{Purity testing}
                };

\node[rootnode]      (var) [left=of avc]                         {\emph{VARIANCE} \\ \small \textsf{Degree-$2$ functions} \\ \small \textsf{Approximate designs}  \\ \small 
                };;

\draw[->] (qsd.north) -- (qqc.south);
\draw[->] (var.north) .. controls +(up:10mm) and +(left:15mm) .. (qsd.west);
\draw[->] (avc.north) -- (qsd.south);
\draw[->] (qqc.north) -- (qsq.south);
\draw[->, dashed] (var.north) .. controls +(up:30mm) and +(left:30mm) .. (qsq.west);

\end{tikzpicture}
\caption{We use two proof techniques, the variance bound and average correlation bound, to bound the quantum statistical dimension $\QSDA$, which in turn  lower bounds the decision problem complexity $\QSD$ which (under a set of assumptions) lower bounds $\QSQ$ learning complexity. The variance method can be sometimes used to bound $\QSQ$ directly, which we discuss in Appendix~\ref{app:yao}.}
\end{figure}

We remark that although, our quantum combinatorial parameters are inspired by the classical works of Feldman et al.~\cite{Feldman13,Feldman16,feldman2017statistical}, proving that they lower bound $\QSQ$ complexity and also giving lower bounds for the corresponding concept class using these parameters is non-trivial and is a key technical contribution of our work. Below, we apply these lower bounds to obtain our learning results.

\paragraph{$\QSQ$ versus noisy $\QPAC$.} We now sketch the proof of Result~\ref{res2}. In quantum learning theory, there are a few well-known function classes that are learnable using quantum examples: parities, juntas, DNF formulas, the coupon collector problem, learning codeword states. It was observed in~\cite{arunachalam2020quantum} that the first three classes are learnable in $\QSQ$ already, primarily because a version of Fourier sampling is implementable in $\QSQ$. In this work we first observe that the coupon collector problem and learning codeword states are also learnable in the $\QSQ$ framework. Simultaneously, there have been a few works that have shown exponential lower bounds for learning using separable measurements~\cite{DBLP:conf/focs/ChenCH021,huang2022quantum,moore2008symmetric}, but all these lower bounds correspond to learning classes of {mixed quantum states}. Prior to our work, it was open if there is very simple structured \emph{function class} such that quantum examples corresponding to this function class is hard for $\QSQ$ (in fact given our polynomial relation between entangled and separable learning, it is conceivable that for the small class of function states, $\QSQ$ are $\QPAC$ are polynomially related as well). In this work, we look at the degree-$2$ concept class $\Cc$ defined in Eq.~\eqref{eq:deg2functionclass}.

Recently it was observed that~\cite{arunachalam2022optimal} this class is learnable using $O(n)$ quantum examples with entangled measurements and $O(n^2)$ quantum examples with separable measurements. Our main contribution is in showing that the $\QSQ$ complexity of learning $\Cc$ with tolerance $\tau$ is $\Omega(2^n\cdot \tau^2)$, in particular implying that a tolerance $\tau=1/\poly(n)$ implies an exponential $\Omega(2^n)$ lower bound. The proof of this lower bound uses the variance lower bounding technique to lower bound $\QSDA$ (and in turn $\QSQ$). The essential idea is as follows: let $\ket{\psi_A}=\frac{1}{\sqrt{2^n}}\sum_x |x,x^\top A x\rangle$, then we can show that for every $M$ with $\|M\|\leq 1$, we have that the variance
$$
\Var(M\psi_A):=\Exp_A [\Tr(M\psi_A)^2]-\Big(\Exp_A [\Tr(M\psi_A)]\Big)^2
$$
is at most $O(2^{-n/2})$. Proving this upper bound is fairly combinatorial but crucially it involves understanding the properties of the ensemble $\{\ket{\psi_A}\}_A$ and its moments $A$ picked uniformly at random. Finally, we observe that the concept class can be learned given noisy quantum examples like in Eq.~\eqref{eq:noisyintroexample} using $\poly(n,1/(1-2\eta))$ examples. This gives us the claim separation between $\QSQ$ and noisy-$\QPAC$, the ``quantum analogue" of the seminal result of Blum et al.~\cite{blum2003noise} for a \emph{natural} class and with \emph{non-constant} error rate close to $1/2$.
  
\paragraph{New $\QSQ$ lower bounds.}
Using our lower bounding technique, we consider fundamental problems in quantum computing and prove $\QSQ$ lower bounds for these~tasks. 

{\emph{Approximate designs}} An application of our variance-based lower bounds shows that learning ensembles of states forming approximate Haar $2$-designs requires $\Omega(\tau^2\cdot 2^n)$ $\Qstat$ queries. This includes interesting ensembles such as stabilizer states, which are known to be efficiently learnable with even separable measurements, and $\poly(n)$-depth random circuits~\cite{Harrow_Low_2009}.

{\emph{Hidden subgroup problem.}} Coset states appear often in the hidden subgroup problem $(\textsf{HSP})$~\cite{kitaev1995quantum,simon1997power}, a fundamental problem in quantum computing. It is well-known that coset states of the Abelian \textsf{HSP} can be learned exactly from \emph{separable} measurements in polynomial sample complexity and for non-Abelian groups, it was well-known that separable measurements~\cite{moore03,hallgren2000normal} require \emph{exponential} many copies to learn a coset state. A natural question is, what is the $\QSQ$ complexity of learning coset states? Given that the \emph{standard approach} for $\textsf{HSP}$ is based of Fourier sampling and~\cite{arunachalam2020quantum} showed that a version of Fourier sampling is easy in $\QSQ$, it is natural to expect that $\textsf{HSP}$ is implementable in $\QSQ$. Surprisingly, in this work, we show that, even for \emph{Abelian} groups, the $\QSQ$ sample complexity of learning the unknown coset state is exponentially large. In particular, we show a lower bound of $\Omega(\tau^2\cdot 2^n)$ on the $\QSQ$ complexity of learning using $\Qstat(\tau)$ queries and the proof of this is done using the average correlation method. Thus, the abelian hidden subgroup problem cannot be solved in $\QSQ$ given access to only coset states.

{\emph{Shadow tomography.}} The past few years have seen a lot of works understanding shadow tomography~\cite{aaronson:shadow,huang2022quantum}.  The goal here is, given copies of an unknown quantum state $\rho$, the learner has to predict the expectation value $\Tr[O_i \rho]$ of a collection of known observables $\{O_i\}_{i\in [k]}$ up to error $\varepsilon$. It is well-known to be solvable using $\poly(n,\log k)$ copies of $\rho$. In \cite{DBLP:conf/focs/ChenCH021} the authors show $\Theta(2^n)$ copies of $\rho$ are necessary and sufficient for shadow tomography using \emph{separable measurements}. To prove the lower bounds the authors construct a many-vs-one decision task where~$\sigma = \id/{2^n}$ and 
$    \Cc = \{\rho_i = \frac{\id+3\varepsilon O_i}{2^n}\}$. 
Assuming that $\Tr[O_i] =0 $ and $\Tr[O_i^2]=2^n$ for all $O_i$, then an algorithm which solves the shadow tomography problem  also solves the decision problem. Thus, a lower bound on the latter is also a lower bound on the sample complexity of shadow tomography. Here we give a quadratically \emph{stronger} lower bound of $\Omega(4^n)$ when given access to only $\Qstat$ measurements, which we prove using the average correlation method. Our result shows that even separable measurements and not just statistics play a non-trivial role in shadow tomography.

{\emph{Does tolerance matter?}} A natural question when discussing $\QSQ$ learning is, is there a natural \emph{distribution learning} task that can be solved with tolerance $\tau_Q\geq \tau_C$ such that classical $\Stat(\tau_C)$ queries cannot solve the task but $\Qstat(\tau_Q)$ can solve the task? Here we consider the class of \emph{bi-clique states} introduced in the seminal work of Feldman et al.~\cite{feldman2017statistical}. In their work they showed that for detecting a planted bipartite $k$-clique distributions when the planted $k$-clique has size $n^{1/2-\varepsilon}$ (for constant $\varepsilon>0$), it is necessary and sufficient to make superpolynomial in $n$ many $\Stat(k/n)$ queries. Here we show that one can achieve the same query complexity quantumly but with $\Qstat(\sqrt{k/n})$, i.e., with quadratically larger tolerance we can detect a $k$-biclique. A classical $\SQ$ algorithm cannot solve this task with $\tau_C = \sqrt{k/n}$ queries.

{\emph{A doubly exponential lower bound?}} So far all our lower bounds for learning $n$-qubit quantum states are exponential in $n$. A natural question is, can one prove a doubly exponential lower bound for some task? In this work, we show that the natural problem of \emph{testing purity}, i.e., given a quantum state $\rho$ return an estimate of $\Tr[\rho^2]$, requires $\exp(2^n\tau^2)$ many $\QSQ$ queries to solve. Previous work of~\cite{DBLP:conf/focs/ChenCH021} showed that it is necessary and sufficient to use $\Theta(2^n)$ many copies of $\rho$ to test purity if we were allowed \emph{separable} measurements, but our work considers the weaker $\QSQ$ model and proves a doubly-exponential lower bound. The proof of this uses Levy's lemma and the ensemble of Haar random states to lower bound the quantum statistical dimension in a manner similar to that of the variance based technique.

{\emph{General search problems}.} While it is beyond the scope of this paper's main goals in showing separations in learning complexity, in Appendix~\ref{app:qssd} we give a combinatorial parameter, based on $\QSDA$, \emph{characterizing} the complexity of general \emph{search} problems.  We remark that the problems we considered above can all be cast as a search problem. We show that this combinatorial parameter both upper and lower bounds the number of $\Qstat$ queries needed to solve a search problem.

\subsection{Open questions}
There are a few natural questions that our work opens up: $(i)$ Can we show that for every concept class $\Cc$, we have that $\textsf{SepExact}\leq O(n\cdot \textsf{EntExact})?$,\footnote{ We remark that for distribution-independent approximate learning, this inequality is true. This uses the result of~\cite{chung2018sample} (Proposition~\ref{prop:CLsep} below). They showed how to learn every concept class $\Cc$ using $O\big((\log |C| )/\varepsilon\big)$ many quantum examples. By Sauer's lemma, we know that $\log |\Cc|\leq n\cdot \textsf{VC}(\Cc)$, so the $\textsf{SepApp}$ upper bound is $O\big(n \textsf{VC}(\Cc)/\varepsilon\big)$. In~\cite{arunachalam2018optimal} it was shown that the quantum entangled sample complexity of $\varepsilon$-PAC learning is $\Omega(\textsf{VC}(\Cc)/\varepsilon)$ for every $\Cc$. Putting together both these bounds proves the inequality.}
 $(ii)$ Following~\cite{hinsche2022single,sweke23} what is $\QSQ$ complexity of learning the output distribution of constant-depth circuits assuming we only use \emph{diagonal} operators? $(iii)$ Theoretically our work separates weak and strong error mitigation, but in practice there are often assumptions in the mitigation protocols, can we show theoretical separations even after making these assumptions? $(iv)$ Classically it is well-known that several algorithms can be cast into the $\QSQ$ framework, is the same true quantumly? If so, that would suggest that $\QSQ$ as a unifying framework for designing new learning algorithms. $(v)$ What is the $\QSQ$ complexity of the Hidden subgroup problem when given access to \emph{function} states, instead of coset states (which is the case only in the \emph{standard approach}).

\paragraph{Acknowledgements.} We thank Ryan Sweke for useful discussions and sharing a preprint of their work~\cite{sweke23}. We also thank the Quantum algorithms group at IBM, Eric Chitambar, and Felix Leditzky for discussions. SA and LS were partially supported by IBM through the IBM-Illinois Discovery Accelerator Institute.

\paragraph{Organization.} In Section~\ref{sec:prelim} we introduce a few useful theorems and all the learning models we will be concerned with in this paper, in Section~\ref{sec:entsep} we prove our polynomial relation between entangled and separable measurements, in Section~\ref{sec:qsqupper} we give a few concept class that can be learned in $\QSQ$ and prove our main theorem which gives a lower bounding technique for $\QSQ$ learning, in Section~\ref{sec:deg2sep} we prove our exponential separation between noisy-$\QPAC$ and $\QSQ$, in Section~\ref{sec:examples} we give further examples of states for which one can show an exponential lower bound for $\QSQ$ learning and finally in Section~\ref{sec:app} we discuss our two applications in error mitigation and learning distributions.


\section{Preliminaries}
\label{sec:prelim}


\subsection{Quantum Information Theory}
Qubits are unit vectors in $\C^2$ with a canonical basis given as $\ket{0} = \begin{pmatrix} 1 \\ 0 \end{pmatrix}$ and $\ket{1} = \begin{pmatrix} 0 \\ 1 \end{pmatrix}$. Pure quantum states composed of $n$ qubits are unit vectors in $(\C^2)^{\otimes n} \cong \C^{2^n}$. Following Dirac notation we indicate a state by $\ket{\psi}$ and its conjugate transpose, an element of $(\C^{2^n})^*$, by $\bra{\psi}$. Outer products are indicated by the notation $\ketbra{\psi}{\phi}$. Mixed states $\rho$ are positive semi-definite linear operators on $\C^{2^n}$ such that $\Tr[\rho] = 1$. Via the spectral theorem, any mixed state can be decomposed into a probability distribution over projectors onto pure states $\rho = \sum_i \lambda_i \ketbra{u_i}{u_i}$ where $\lambda_i \geq 0$ and $\sum_i \lambda_i = 1$. Pure states correspond to rank $1$ mixed states. We will often label the mixed state corresponding to a pure state $\ket{\psi}$ by $\psi$ (instead of $\ketbra{\psi}{\psi}$). Note that all states in the orbit of a pure states under phases $e^{i\theta})$ correspond to the same mixed state. Thus, $\ket{\psi}$ and $e^{i\theta}\ket{\psi}$ are considered to be the same state. Positive operator valued measures (POVMs) capture the most general notion of quantum measurements. These are given by ensembles of positive semi-definite operators $\{E_i\}_i$ such that $\sum_i E_i = \id$. The probability of measurement outcome $i$ is given by $\Tr[E_i \rho]$. Observables $M$ are bounded Hermitian operators on $\C^{2^n}$, representing a measurement with values assigned to the outcomes. The expectation value of an observable is given by $\Tr[M\rho]$. Quantum computers operate by applying gates (unitary matrices) to (ideally) pure states which evolve like $\ket{\psi'} = U\ket{\psi}$ where $U$ is some unitary. One important gate we will see is the Hadamard gate, given by $H = \frac{1}{\sqrt{2}}\begin{pmatrix} 1 & 1 \\ 1 & -1 \end{pmatrix}$.

\subsection{Notation}
Throughout for $n\geq 1$, we let $[n]=\{1,\ldots,n\}$. For quantum states $\ket{\psi},\ket{\phi}$, we denote $\TRD(\ket{\phi},\ket{\psi})$ as the trace distance between the states $\ket{\phi}$ and $\ket{\psi}$, which, for pure states, is defined as
$$
\TRD(\ket{\phi},\ket{\psi})=\sqrt{1-|\langle \phi|\psi\rangle|^2},
$$
and for mixed states $\rho,\sigma$, define $\TRD(\rho,\sigma)=\frac{1}{2}\|\rho-\sigma\|_1$ (where $\|\cdot \|_1$ is the matrix Schatten-$1$ norm). For mixed states $\rho,\sigma$, the operational interpretation of the trace distance is given by
$$
\TRD(\rho,\sigma)=\frac{1}{2}\max_{M: \| M \| \leq 1} | \Tr(M (\rho - \sigma))|.
$$
For distributions, $P,Q:\01^n\rightarrow [0,1]$, we say $\Pr_{x\sim P}$ to mean $x$ is sampled from $P$. We indicate sampling $x$ uniformly from some set $\X$ by $x\sim \X$. For example, $x\sim \01^n$ and $\rho \sim \Cc$ respectively denote sampling a bitstring or a state from an ensemble uniformly at random. Similarly, we say $\TVD(P,Q)$ to mean the total-variational distance between $P,Q$ defined as $\TVD(P,Q)=\frac{1}{2}\sum_x |P(x)-Q(x)|$. Similarly, define the Hellinger distance between $P,Q$ as
$$
\HD(P,Q)^2=1-\Big(\sum_x \sqrt{P(x) Q(x)}\Big)^2.
$$

\subsection{Useful theorems}

\begin{theorem}[Levy's lemma]
\label{thm:levy}
Let $f:S^{d-1}\rightarrow \C$ be a function on the $d$-dimensional unit sphere $S^{d-1}$. Let $k$ be such that for every $\ket{\phi},\ket{\psi}\in S^{d-1}$, we have that
$$
|f(\ket{\psi})-f(\ket{\phi})|\leq k\cdot \|\phi-\psi\|_2,
$$
then there exists a constant $C>1$ such that 
$$
\Pr\Big[|f(\psi)-\Exp[f(\psi)]|\geq \varepsilon\Big]\leq 2\exp(-Cd \varepsilon^2/k^2),
$$
where the probability and expectations are over the Haar measure on $S^{d-1}$.
\end{theorem}

\begin{fact}
\label{fact:2statedistinguishingbound}
Let binary random variable $\mathbf{b}\in\01$ be uniformly distributed. Suppose an algorithm is given $\ket{\psi_{\mathbf{b}}}$ (for unknown $b$) and is required to guess whether $\mathbf{b}=0$ or $\mathbf{b}=1$. It will guess correctly with probability at most $\frac{1}{2}+\frac{1}{2}\sqrt{1-|\braket{\psi_0}{\psi_1}|^2}$.

Note that if we can distinguish $\ket{\psi_0},\ket{\psi_1}$ with probability $\geq 1-\delta$, then $|\langle\psi_0,\psi_1\rangle|\leq 2\sqrt{\delta(1-\delta)}$.
\end{fact}

\begin{fact}
\label{fact:deg2}
    The class of degree-$2$ phase states $\{\frac{1}{\sqrt{2^n}}\sum_x (-1)^{x^\top A x}\ket{x}:A\in \mathbb{F}_2^{n\times n}\}$ can be learned using $O(n)$ entangled measurements in time $O(n^3)$.
\end{fact}
\begin{proof}
The learning algorithm uses the Bell-sampling procedure: given two copies of $\ket{\phi_A}=\frac{1}{\sqrt{2^n}}\sum_{x} (-1)^{x^\top A x}\ket{x}$, perform $n$ CNOTs between the first copy and  second copy and measure the second register to obtain a uniformly random $y\in \mathbb{F}_2^n$.  The resulting quantum state is
$$
\frac{1}{\sqrt{2^n}}\sum_x (-1)^{x^\top A x+(x+y)^\top A(x+y)}\ket{x}=\frac{(-1)^{y^\top Ay}}{\sqrt{2^n}}\sum_x (-1)^{x^\top(A+A^\top)\cdot y}\ket{x}.
$$
The learning algorithm then applies the $n$-qubit Hadamard transform and measures to obtain bit string $(A+A^\top)\cdot y$. Repeating this process $O(n)$ many times, one can learn $n$ linearly independent constraints about $A$. Using Gaussian elimination, this procedure allows the learner to learn the off-diagonal elements of $A$. {In order to learn the diagonal elements of $A$ a learning algorithm applies the operation $\ket{x}\rightarrow (-1)^{x_{i}\cdot x_j}\ket{x}$  if $A_{ij}=1$ for every $i\neq j$. The resulting quantum state is $\sum_x (-1)^{\sum_i x_i A_{ii}}\ket{x}$ and the learner can apply the $n$-qubit Hadamard transform to learn the diagonal elements of $A$.}  
\end{proof}
\begin{fact}
\label{fact:traceandhellinger}
    For distributions $p,q:\mathcal{X}\rightarrow [0,1]$, define $\ket{\psi_p}=\sum_{x\in \mathcal{X}} \sqrt{p(x)}\ket{x}$ and $\ket{\psi_q}$ similarly.~Then $$\TRD(\ket{\psi_p},\ket{\psi_q})^2 \leq 2\TVD(p,q).
    $$
\end{fact}

\begin{proof}
In order to see the fact, first we have that
    $$
\TRD(\ket{\psi_p},\ket{\psi_q})^2 =1-\langle\psi_p|\psi_q\rangle^2=1-\Big(\sum_x \sqrt{p(x)q(x)}\Big)^2.
$$
By the definition of the Hellinger distance, we have that $d_H(p,q)^2 =1-\sum_x \sqrt{p(x)q(x)}$, so we have
\begin{align*}
\TRD(\ket{\psi_p},\ket{\psi_q})^2 & =2\big(1- \sum_x \sqrt{p(x)q(x)} \big)-\big(1-\sum_x \sqrt{p(x)q(x)}\big)^2=2d_H(p,q)^2-d_H(p,q)^4\\
& \leq 2d_H(p,q)^2\leq 2\TVD(p,q), 
\end{align*}

where the final inequality used~\cite[Proposition~1]{daskalakis2018distribution}.
\end{proof}

\begin{fact}
    \label{lem:statetodistribution}
    For a distribution $p:\01^n\rightarrow [0,1]$, let $\ket{\psi_p}=\sum_x \sqrt{p(x)}\ket{x}$.  Suppose there exists an algorithm that makes $t$ $\Qstat$  queries and learns $p$ up to total variation distance $\varepsilon^2$, then there exists an algorithm  that makes $t$ $\Qstat$ queries and learns $\ket{\psi_p}$ up to trace distance $\sqrt{2}\varepsilon$.
\end{fact}
\begin{proof}
By Fact~\ref{fact:traceandhellinger}, first observe that
$
\TRD(\ket{\psi_p},\ket{\psi_q})^2_{\Tr}\leq 2\TVD(p,q). 
$
  Now the lemma statement follows immediately: suppose there exists an algorithm that makes $\Qstat$ queries to $\ket{\psi_p}$ and outputs a $q$ such that $\TVD(p,q)\leq \varepsilon^2$, then that implies that $\TRD(\ket{\psi_p}, \ket{\psi_q}) \leq \sqrt{2}\varepsilon$.
\end{proof}

\begin{fact}[Discriminating coherent encodings of distributions]
\label{Lemma:norms}
For distributions $D, D_0$ over some domain $X$,  let $\ket{\psi} =  \sum_{x \in X} \sqrt{D(x)} \ket{x}$ and $\ket{\psi_0} = \sum_{x \in D_0} \sqrt{D_0(x)}\ket{x}$. We have that
$$
\max_{\phi: X \rightarrow [-1,1] } \left|\sum_x (D(x) - D_0(x))\phi(x)\right| = 2\TVD(D, D_0).
$$ 
\end{fact}
\begin{proof}
 Choose $\phi = \delta(D(x) > D_0(x)) - \delta (D_0(x) \leq D(x))$, where $\delta (\cdot)$ is an indicator function from Boolean clauses to $\lbrace 1, 0 \rbrace$ which evaluates to $1$ if its argument evaluates to true and evaluates to $0$ if its argument is false.
    \begin{equation}
    \begin{aligned}
     \left|\sum_{x \in X} (D(x) - D_0(x)) \phi(x) \right| &= \sum_{x \in X; D(x) > D_0(x)} (D(x) - D_0(x)) + \sum_{x \in X; D(x) \leq D_0(x)} (D_0(x) - D(x))\\
     &= \sum_{x \in X} |D(x) - D_0(x)| = 2 \TVD(D, D_0),
     \end{aligned}
    \end{equation}
    hence proving the fact.
\end{proof}

 \begin{fact}
 \label{fact:zippel}
        For distinct $A,B\in \mathbb{F}_2^{n\times n}$, we have that
        $
            \Pr_{x\sim \01^n}[x^\top Ax \neq x^\top Bx] \geq 1/4.
        $
    \end{fact}
    \begin{proof}
        $Pr_{x\sim \01^n}[x^\top Ax \neq x^\top Bx]  = \Pr_{x\sim \01^n}[x^\top Ax \oplus x^\top Bx \neq 0] \geq \frac{1}{4},$ where the inequality follows from the Schwartz-Zippel lemma for Boolean functions \cite{nisan_szegedy_1994}.
    \end{proof}

\subsection{Learning models}
In this section we first describe the learning models we will be concerned with in this paper.

\paragraph{Classical $\PAC$ learning.} Valiant~\cite{DBLP:journals/cacm/Valiant84} introduced the classical Probably Approximately Correct ($\PAC$) learning model. In this model,  a \emph{concept class} $\Cc\subseteq \{c:\01^n\rightarrow \01\}$ is a collection of Boolean functions. The learning algorithm $\A$ obtains \emph{labelled examples} $(x,c(x))$ where $x\in \01^n$ is uniformly random and $c\in \Cc$ is the \emph{unknown} {target} function.\footnote{More generally in $\PAC$ learning, there is an unknown distribution $D:\01^n\rightarrow [0,1]$ from which $x$ is drawn.  Throughout this paper we will be concerned with uniform-distribution $\PAC$ learning, i.e., $D$ is the uniform distribution, so we describe the learning model for the uniform distribution for simplicity.} The goal of an $(\varepsilon,\delta)$-learning algorithm $\A$ is the following: for every $c\in \Cc$, given labelled examples $\{(x^i,c(x^i))\}_i$, with probability $\geq 1-\delta$ (over the randomness of the labelled examples and the internal randomness of the algorithm), output a hypothesis $h:\01^n\rightarrow \01$ such that $\Pr_x [c(x)=h(x)]\geq 1-\varepsilon$. The $(\varepsilon,\delta)$-sample complexity of a learning algorithm $\A$ is the maximal number of labelled examples used, maximized over all $c\in \Cc$. The $(\varepsilon,\delta)$-sample complexity of learning $\Cc$ is the minimal sample complexity  over all $(\varepsilon,\delta)$-learners for $\Cc$. Similarly the $(\varepsilon,\delta)$-time complexity   of learning $\Cc$ is the total number of time steps used by an optimal $(\varepsilon,\delta)$-learner for $\Cc$.

\paragraph{Quantum $\PAC$ learning.}  The quantum $\PAC$ was introduced by Bshouty and Jackson~\cite{DBLP:conf/colt/BshoutyJ95} wherein, they allowed the learner access to quantum examples of the form
$$
\ket{\psi_c}=\frac{1}{\sqrt{2^n}}\sum_{x\in \01^n}\ket{x,c(x)}.
$$
Note that measuring $\ket{\psi_c}$ in the computational basis produces a classical labelled example, so quantum examples are at least as strong as classical examples. Understanding their strength and weakness has been looked at by several works (we refer an interested reader to the survey~\cite{arunachalam2017guest}). Like the classical complexities, one can similarly define the $(\varepsilon,\delta)$-sample and time complexity for learning $\Cc$ as the quantum sample complexity (i.e., number of quantum examples $\ket{\psi_c}$) used and quantum time complexity (i.e., number of quantum gates used in the algorithm) of an optimal $(\varepsilon,\delta)$-learner for $\Cc$.

\paragraph{Quantum $\PAC$ learning with classification noise.}  Classically, the $\eta$-classification noise model is defined as follows: for an unknown $c\in \Cc$, a learning algorithm is given uniformly random $x\in \01^n$ and $b\in \01$ where, $b=c(x)$ with probability $1-\eta$ and $b=\overline{c(x)}$ with probability $\eta$. In the same work, Bshouty and Jackson~\cite{DBLP:conf/colt/BshoutyJ95} defined quantum learning with classification noise, wherein a learning algorithm is given access to 
$$
\ket{\psi^n_c}=\frac{1}{\sqrt{2^n}}\sum_{x\in \01^n}\ket{x} \otimes (\sqrt{1-\eta}\ket{c(x)}+\sqrt{\eta}\ket{\overline{c(x)}}).
$$
Such quantum examples have been investigated in prior works~\cite{DBLP:conf/colt/BshoutyJ95,arunachalam2018optimal,grilo2019learning}.

\paragraph{Learning with entangled and separable measurements.} Observe that in the usual definition of $\QPAC$ above, a learning algorithm is given access to $\ket{\psi_c}^{\otimes T}$ and needs to learn the unknown $c\in \Cc$. In this paper we make the distinction between the case where the learner uses entangled measurements, i.e., perform an arbitrary operation on copies of $\ket{\psi_c}$ versus the setting where the learner uses separable measurements, i.e., performs a single-copy measurement on every copy of $\ket{\psi_c}$ in the learning algorithm. When discussing learning with entangled and separable measurements, in this paper we will be concerned with \emph{exact} learning, i.e., with probability $\geq 2/3$, the learner needs to \emph{identify} $c$. We denote $\textsf{EntExact}$ as the  sample complexity of learning with entangled measurements and $\textsf{SepExact}$ as the  sample complexity of learning with separable measurements.
  
\paragraph{Quantum statistical query learning.}
We now discuss the $\QSQ$ model, following the definitions given in~\cite{arunachalam2020quantum}.  We first discuss the classical statistical query ($\SQ$) model for learning an unknown concept  $c^\star \in \Cc$. Classically, the learner has access to a \emph{statistical query oracle} $\Stat$, that on input a function $\phi:\01^{n+1}\rightarrow [0,1]$ and a \emph{tolerance} $\tau$ and returns a number $\alpha$ satisfying
	\[  
	\Big |\alpha - \mathop{\Exp}_{x\sim \01^n}[\phi(x,c^*(x))] \Big | \leq \tau\;.
	\]
A classical $\SQ$ algorithm can adaptively make a sequence of $\Stat$ queries $\{(\phi_i,\tau_i)\}_i$ and based on the responses $\{\alpha_i\}$, with probability $\geq 1-\delta$ it outputs a hypothesis $h:\01^n\rightarrow \01$. The goal of the classical $\SQ$ algorithm is to  output an $h$ such that $\Pr_x [h(x)=c(x)]\geq 1-\varepsilon$. The query complexity of a classical $\SQ$ algorithm is the number of $\Stat$ queries the algorithm makes and the time complexity is the total number of gates used by the algorithm and in the description of the~hypothesis. 

 A natural way to extend the learning model is to allow the algorithm \emph{quantum statistical queries}.  In the classical case, one can think of the input $\phi$ to the $\Stat$ oracle as a specification of a \emph{statistic} about the distribution of examples $(x,c^*(x))$, and the output of the $\Stat$ oracle is an  {\em estimation} of $\phi$: one can imagine that the oracle receives i.i.d.~labeled examples $(x,c^*(x))$ and empirically computes an estimate of $\phi$, which is then forwarded to the learning algorithm.	In the quantum setting, one can imagine the analogous situation where the oracle receives copies of the quantum example state $\ket{\psi_{c^*}}$, and performs a measurement  indicated by the {\em observable} $M$ on each copy and outputs an estimate of $\langle \psi_{c^*}| M | \psi_{c^*}\rangle$. 

 Relaxing the assumption of example states, in order to learn an unknown (mixed) quantum state $\rho$ in the $\QSQ$ model the learner makes $\Qstat$ queries that takes as input an operator $M\in \C^{2^{n+1}\times 2^{n+1}}$ and tolerance $\tau$ and outputs a $\tau$-approximation of $\Tr(M\rho)$, i.e.,
$$
\Qstat: (M,\tau)\mapsto \alpha \in [\Tr(M\rho)+\tau,\Tr(M\rho)-\tau].
$$
In order to learn the concept class using quantum examples, we define $\rho=\ketbra{\psi_c}{\psi_c}$, so the action of the $\Qstat$ oracle is defined as
$$
\Qstat: (M,\tau)\mapsto \alpha \in [\langle \psi_c|M|\psi_c\rangle+\tau, \langle\psi_c|M|\psi_c\rangle-\tau].
$$

In this case, the goal of a $\QSQ$ learner is to output a hypothesis quantum state $\sigma$ that  satisfies $d_{\Tr}(\rho,\sigma)\leq \varepsilon$. If $\rho = \psi_f$ and $\sigma = \psi_h$ this translates to $\Pr_{x \sim \01^n}[f(x) = h(x) ] \geq 1 - \sqrt{\eps}$. Thus, without loss of generality we will often talking about learning states with respect to trace distance, even for learning example states. Our results generally do not depend on assuming that the learner outputs an example state even when the concept class is composed of example states. Clearly, if the learning problem is hard without such a restriction, it is no easier with such a restriction.

We emphasize that the learning algorithm is still a \emph{classical} randomized algorithm and only receives statistical estimates of measurements on quantum~examples. The  quantum query complexity of the $\QSQ$ algorithm is the number of $\Qstat$ queries the algorithm makes and the quantum time complexity is the total number of gates used by the algorithm and in the description of the hypothesis. There are three ways to motivate the $\QSQ$ model
\begin{enumerate}
    \item Clearly any binary measurement $\{M, \Id -M\}$ can be simulated with a $\Qstat$ query to $M$ or $\Id - M$. In the opposite direction, any observable $M$ such that $\Vert M \Vert \leq 1$ can be converted into the POVM $\{\frac{\Id+M}{2}, \frac{\Id - M}{2}\}$. Thus, a $\Qstat$ query is essentially the same model as approximately sampling from a binary POVM up to total variational distance $\Theta(\tau)$. One can think of $\QSQ$ as a sort of noisy variant of binary measurements. From a theoretical perspective, performing noisy $2$-outcomes separable measurements are weak (and easier to implement) than arbitrary separable measurements, which are in turn weaker (and easier to implement) than entangled measurements. So, it is useful to understand the power of such noisy measurements in quantum learning theory and  $\QSQ$ captures this question in a theoretical~framework.

    \item One could envision a situation where quantum states $\rho$ are prepared in the ``cloud" and the \emph{classical} learning algorithm needs to only interact with the cloud \emph{classically}. An efficient $\QSQ$ model allows a quantum advantage in learning in this framework.
    \item The $\QSQ$ model naturally extends recent works~\cite{quek2022exponentially,hinsche2022single,sweke23} wherein they consider the limitations of classical $\SQ$ algorithms for learning a quantum state $\psi_U=U\ket{0^n}$, i.e. they consider  the model where $M$ is diagonal specifiable as $M=\sum_x \phi(x)\ketbra{x}{x}$, then 
    $$
    \langle \psi_U|M|\psi_U\rangle=\sum_x \phi(x)\langle x|U|0^n\rangle^2=\sum_x \phi(x)P_U(x)=\E_{x\sim P_U}[\phi(x)],
    $$
    which is precisely $\alpha_\phi$ they assume access to, in order to learn the unknown $U$.
\end{enumerate}
\emph{Throughout this paper}, for notational convenience we use the following notation: $(i)$ for an $n$-bit problem, when we do not specify a tolerance for the $\Qstat$ oracle, we implicitly assume that the tolerance is $\tau=1/\poly(n)$, $(ii)$  we always make $\Qstat$ queries with an operator $M$ that satisfies $\|M\|\leq 1$, so we do not explicitly state this when discussing $\Qstat$ queries, $(iii)$ We say a $n$-bit concept class $\Cc$ is $\QSQ$ learnable if $\Cc$ can be learned using $\poly(n)$ many $\Qstat$ queries, each with tolerance $\tau=1/\poly(n)$ and observable $M$ which is implementable using $\poly(n)$ many gates.


\section{Relating separable and entangled measurements}
\label{sec:entsep}

Instead of considering just function states over the uniform distribution, we now take the quantum algorithm to have access to states of the form
\begin{align}
    \ket{\psi_c} & = \sum_x \sqrt{D(x)}\ket{x,f(x)}\ ,
\end{align}
where $D$ is some distribution. We further assume that the support of $D$ is the full space $\{0,1\}^n$, to ensure that $\ket{\psi_{c_1}} = \ket{\psi_{c_2}}$ implies that $c_1 = c_2$.

Before proving our main theorem, we will use the following proposition, which was proven earlier in~\cite{chung2018sample} in the context of learning quantum channels. We restate their proposition in the context of learning pure states using parameters that suit our application.

\begin{proposition}
\label{prop:CLsep}
    Let $\mathcal{C}\subseteq \{c:\01^n\rightarrow \01\}$ and $\varepsilon>0$. Given
    \begin{align}
        T=O\left(\frac{\log \vert \mathcal{C} \vert + \log 1/\delta}{\varepsilon}\right)
    \end{align}
    copies of $\ket{\psi_c} = \sum_x \sqrt{D(x)} \ket{x, c(x)}$ for an unknown $c\in \Cc$, there exists an algorithm that uses separable measurements and, with probability $\geq 1-\delta$, outputs a $c'\in \Cc$ such that $\Pr_x [c(x)=c'(x)]\geq 1-\varepsilon$.
\end{proposition}
Note that the proposition in \cite{chung2018sample} deals with general states and they require $T = O\Big(\frac{\log \vert \mathcal{C} \vert + \log 1/\delta}{\eps^2}\Big)$ copies of $\ket{\psi_c}$, where $\eps$ is now error with respect to trace distance. For our purposes the factor of $1/{\eps^2}$ is improved to $1/{\eps}$ by the fact that $\TRD(\psi_f,\psi_h) = \sqrt{\eps}$ implies that $\Pr_x[f(x)\neq h(x)] = \Theta(\eps)$.

\begin{theorem}
    Let $\Cc$ be a concept class $\Cc\subseteq \{c:\01^n\rightarrow \01\}$ and 
    $$
    \eta_{\textsf{m}}=\min_{c,c'\in \Cc}\Pr_x[c(x)\neq c'(x)], \quad    \eta_{\textsf{a}}=\mathop{\Exp}_{c,c'\in \Cc}\Pr_x[c(x)\neq c'(x)].
    $$Then we have that 
$$
\textsf{SepExact}(\Cc)\leq  O\Big(n\cdot \textsf{EntExact}(\Cc)\cdot \min\big\{\eta_{\textsf{a}}/\eta_{\textsf{m}}\ , \ \textsf{EntExact}(\Cc)\Big\}\Big).
$$ Furthermore, there exists $\Cc$ for which this inequality is tight.
\end{theorem}

\begin{proof}
First observe that
    \begin{align}
    \label{eq:seplowerbound1}
\textsf{SepExact}\leq 2/\eta_m\cdot \log |\Cc|.
    \end{align}
This is easy to see: fix $\varepsilon=\eta_m/2$ in Proposition~\ref{prop:CLsep} and consider a separable approximate algorithm that, given copies of $\ket{\psi_c}$,  an approximate learning algorithm outputs $c'$ such that $\Pr_x [c(x)\neq c'(x)]\leq \varepsilon$, then $c=c'$ by definition of $\eta_m$, hence this algorithm is a separable \emph{exact} learner.

Next, we prove two lower bounds on $\textsf{EntExact}$: 
\begin{align}
\label{eq:entlowerbound}
\textsf{EntExact}\geq \max\left\{\frac{1}{\eta_m},\frac{\log|\Cc|}{n\eta_a}\right\}
\end{align}
To see the first lower bound in Eq.~\eqref{eq:entlowerbound}, observe the following: consider the $c,c'\in \Cc$ for which $\Pr_x[c(x)\neq c'(x)]=\eta_m$, then every exact learning algorithm  needs to distinguish between $c,c'$. Since $\langle\psi_c |\psi_{c'}\rangle=1-\eta_m$, by Fact~\ref{fact:2statedistinguishingbound}, this implies a lower bound of $T=\Omega(1/\eta_m)$ many quantum examples to distinguish between $c,c'$ with bias $\Omega(1)$.

To see the second lower bound in Eq.~\eqref{eq:entlowerbound}, first note that $1-\eta_a=\Exp_{c,c'\in \Cc}\Exp_x[c(x)= c'(x)]$. Next, observe that 
 \begin{align}
 \label{eq:entupperbound1}
 \textsf{EntExact}\geq \frac{\log|\Cc|}{n\eta_a}.
 \end{align}
 The proof of this is similar to the information-theoretic proof in~\cite{arunachalam2018optimal}. 		We prove the lower bound for $\Cc$ using a three-step information-theoretic technique.   Let $\mathbf{A}$ be a random variable that is uniformly distributed over $\Cc$. Suppose $\mathbf{A}=c_V$, and let $\mathbf{B}=\mathbf{B}_1\ldots\mathbf{B}_T$ be $T$ copies of the quantum example 
 $$
 \ket{\psi_c}= \sum_{x\in \01^{n}}\sqrt{D(x)} \\ket{x,c(x)}
 $$ 
 for $c\in \Cc$. The random variable $\mathbf{B}$ is a function of the random variable~$\mathbf{A}$. 
		The following upper and lower bounds on $I(\mathbf{A}:\mathbf{B})$ are similar to~\cite[Theorem~12]{arunachalam2018optimal} and we omit the details of the first two steps here.
		\begin{enumerate}
			\item $I(\mathbf{A}:\mathbf{B})\geq \Omega(\log|\Cc|)$ because $\mathbf{B}$ allows one to recover $\mathbf{A}$ with high probability.
			\item $I(\mathbf{A}:\mathbf{B})\leq T\cdot I(\mathbf{A}:\mathbf{B}_1)$ using a chain rule for mutual information.
			
			\item $I(\mathbf{A}:\mathbf{B}_1)\leq O(n\cdot \eta_a)$.\\[1mm]
			\emph{Proof (of 3).} Since $\mathbf{A}\mathbf{B}$ is a classical-quantum state, we have 
			$$
			I(\mathbf{A}:\mathbf{B}_1)= S(\mathbf{A})+S(\mathbf{B}_1)-S(\mathbf{A}\mathbf{B}_1)=S(\mathbf{B}_1),
			$$ 
			where the first equality is by definition and the second equality uses $S(\mathbf{A})=\log |\Cc|$ since $\mathbf{A}$ is uniformly distributed over~$\Cc$, and $S(\mathbf{A}\mathbf{B}_1)=\log |\Cc|$ since the matrix 
			$$
			\sigma=\frac{1}{|\Cc|} \sum_{c\in \Cc} \ketbra{c}{c}\otimes \ketbra{\psi_c}{\psi_c}
			$$ is block-diagonal with $|\Cc|$ rank-1 blocks on the diagonal. It thus suffices to bound the entropy of the (vector of singular values of the) reduced state of $\mathbf{B}_1$, which~is
			$$
			\rho=\frac{1}{|\Cc|}\sum_{c\in \Cc}\ketbra{\psi_c}{\psi_c}.
			$$
			Let $\sigma_0\geq \sigma_1\geq\cdots\geq \sigma_{2^{n+1}-1}\geq 0$ be the singular values of $\rho$. Since~$\rho$ is a density matrix, these form a probability distribution. Now observe that $\sigma_0\geq 1-\eta_a$: consider the vector $u=\frac{1}{|\Cc|}\sum_{c'\in \Cc}\ket{\psi_{c'}}$  and observe that
   \begin{align*}
  u^\top \rho u &=\frac{1}{|\Cc|^3}\sum_{c,c',c''\in \Cc}\langle \psi_c|\psi_{c'}\rangle\langle \psi_c|\psi_{c''}\rangle\\
  &=\Exp_{c} \Big[\Exp_{c'}[\langle\psi_c|\psi_{c'}\rangle]\Big]\cdot\Big[\Exp_{c''}[\langle \psi_c|\psi_{c''}\rangle]\Big]\\
  &\geq \Big(\mathop{\Exp}_{c,c'}[\langle\psi_c|\psi_{c'}\rangle]\Big)\cdot \Big(\mathop{\Exp}_{c,c''}[\langle \psi_c|\psi_{c''}\rangle]\Big)=\big(\mathop{\Exp}_{c,c'\in \Cc}\Pr_x[c(x)=c'(x)]\big)^2\geq 1-2\eta_a,
   \end{align*}
where the first inequality is by Chebyshev's sum inequality (since all the inner products are non-negative) and the second  inequality  follows from the definition of $\eta_a$. Hence we have that $\sigma_0=\max_{u}\{u^\top \rho u / u^\top u\} \geq 1-2\eta_a$ (where we used that $\|u\|_2\leq 1$). 

   Let $\mathbf{N}\in\{0,1,\ldots,2^{n+1}-1\}$ be a random variable with probabilities $\sigma_0,\sigma_1,\ldots,\sigma_{2^{n+1}-1}$, and $\mathbf{Z}$ an indicator for the event ``$\mathbf{N}\neq 0$.'' Note that $\mathbf{Z}=0$ with probability $\sigma_0\geq 1-2\eta_a$, and $H(\mathbf{N}\mid \mathbf{Z}=0)=0$. By a similar argument as in~\cite[Theorem~15]{arunachalam2018optimal}, we~have 
			\begin{align*}
			S(\rho) & =H(\mathbf{N})=H(\mathbf{N},\mathbf{Z})=H(\mathbf{Z})+H(\mathbf{N}\mid\mathbf{Z})\\
			& =H(\sigma_0)+\sigma_0\cdot H(\mathbf{N}\mid \mathbf{Z}=0) + (1-\sigma_0)\cdot H(\mathbf{N}\mid \mathbf{Z}=1) \\
			& \leq H(\eta_a) + \eta_a(n+1)\\
   &\leq O(\eta_a(n+\log (1/\eta_a)) 
			\end{align*}
			using $H(\alpha)\leq O(\alpha\log (1/\alpha))$.
		\end{enumerate}
		Combining these three steps implies $T=\Omega(\log |\Cc| / (n\eta_a))$.  Now putting the relations between $\textsf{EntExact},\textsf{SepExact}$ together we get
$$
\textsf{SepExact}\leq n\cdot \eta_a/\eta_m\cdot \textsf{EntExact}\leq  n\cdot \textsf{EntExact}^2,
$$
hence we have the desired upper bound as in the theorem statement\footnote{We state the theorem as below, since it is apriori unclear as to why $1/\eta_m$ is a lower bound on $\textsf{EntExact}$.}
$$
\textsf{SepExact}\leq  O\Big(n\cdot \textsf{EntExact}\cdot \min\left\{\eta_{\textsf{a}}/\eta_{\textsf{m}},\textsf{EntExact}\right\}\Big).
$$
 To show that this inequality is optimal, observe that:  if $\Cc$ is the class of degree-$2$ phase states, i.e., $\Cc=\{f_A(x)=x^\top A x:A\in \01^{n\times n}\}$, then $\eta_m=\eta_a= \Theta(1)$ by Fact~\ref{fact:zippel}. We saw in Fact~\ref{fact:deg2} that this class can be learned using $\Theta(n)$ entangled measurements, so $\textsf{EntExact}= \Theta(n)$ and the above upper bound implies  $\textsf{SepExact}=O(n^2)$, which was shown to be optimal in~\cite{arunachalam2022optimal}.
\end{proof}


\section{Lower bounds for Quantum statistical query learning}
\label{sec:qsqupper}

Here we prove our main theorem which provides combinatorial quantities one can use to lower bound the $\QSQ$ complexity of various tasks. The techniques and parameters used  in this section are inspired by several seminal classical works on the classical $\SQ$ model~\cite{kearns1998efficient,Feldman13,Feldman16,feldman2017statistical}. We first define \emph{statistical decision problems}, then define the quantum statistical dimension ($\QSDA$) which lower bounds the decision problem complexity and finally discuss the variance and average correlation lower bounds on $\QSDA$. For notational convenience we adopt the following shorthand: we will let $\Cc$ be a collection of $n$-qubit quantum states. We let $\QSQ_\tau^{\eps,\delta}(\Cc)$ be the complexity of learning $\Cc$ to accuracy $\eps$ in trace distance using $\Qstat$ queries of tolerance $\tau$ and succeeding with probability at least $1-\delta$. Next, $\QSD_\tau^\delta(\Cc, \sigma)$ is the complexity of \emph{deciding} if $\rho \in \Cc$ or $\rho = \sigma$ given $\Qstat(\tau)$ access to $\rho$, $\QSDA$ is the quantum statistical dimension, and $\QAC$ is the average correlation~bound.

\subsection{Learning is as hard as deciding}

\begin{definition}[Quantum (many-vs-one) Decision Problem]
     Let $\tau \in [0,1]$ and let $\sigma \not\in \Cc$. A quantum statistical decision problem for $(\Cc, \sigma)$ is defined as: for an unknown state $\rho$, given $\Qstat(\tau)$ access to $\rho$ decide if $\rho\in \Cc$ or $\rho=\sigma$. Let $\QSD_\tau^\delta(\Cc, \sigma)$ be the number of $\Qstat(\tau)$ queries made by the best algorithm for the decision problem that succeeds with probability at least $1-\delta$.
     \label{Dfn:statistical_decision_problem}
\end{definition}

We now prove our first lemma that $\QSD$ is actually a lower bound on $\QSQ$ learning the concept class $\Cc$. We remark that a similar lemma appears for \emph{classical} $\SQ$ in~\cite{sweke23}, we want to thank the authors for discussing and sharing their manuscript during the completion of our work.

\begin{lemma}[Learning is at least as hard as deciding]
\label{Lemma:LearningToDeciding}
Let $\varepsilon \geq \tau > 0$ and  $\sigma\notin \Cc$ be such that $\min_{\rho \in \Cc} [\TRD(\rho, \sigma)] > 2(\tau + \varepsilon)$.   Let $\QSQ_\tau^{\varepsilon, \delta}(\Cc)$ be the number of $\Qstat(\tau)$ queries made by a $\QSQ$ algorithm that on input $\rho$ outputs $\pi$, such that $\TRD(\pi, \rho) \leq \varepsilon$ with probability $\geq 1-\delta$. Then 
    $$
    \QSQ_\tau^{\varepsilon, \delta}(\Cc) \geq \QSD_\tau^\delta(\Cc,\sigma) - 1.
    $$ 
\end{lemma}

\begin{proof} We show this by solving the statistical quantum decision problem by querying a $\QSQ_\tau^{\varepsilon, \delta}(\Cc)$ learning algorithm $\A$. For $\rho \in \mathcal{C}$, $\A$ outputs, with probability $\geq 1-\delta$, a classical description of quantum state $\pi$ such that $\TRD(\rho, \pi) \leq \varepsilon$. Note that for $\sigma \notin \mathcal{C}$, the output of $\A$ is not well-defined and we assume that $\A$ can output anything. 

Let the output of $\A$ be $\pi$. If $\pi$ is not a valid quantum state, return ``$\rho = \sigma$''.  We then check if $\min_{\rho \in \Cc} \TRD(\pi, \rho) > \varepsilon$. If yes, return ``$\rho = \sigma$''. At this point, we know the classical description of both $\pi$ and $\sigma$ and also know that there exists some $\nu \in \Cc$, such that $\TRD(\pi, \nu) \leq \varepsilon$. We can find such $\nu$ that is closest to $\pi$, as well as an operator $\Pi_+$ which is a projector onto the positive part of the spectrum of the hermitian operator $\nu-\sigma$. Finding this may be computationally difficult, but does not require additional $\Qstat(\tau)$ queries. We then query $\Qstat(\tau)$ with $\Pi_+$ to obtain a response $R$. If $| R - \Tr(\Pi_+ \sigma) | \leq \tau$, return ``$\rho = \sigma$''. Return ``$\rho \in \Cc$'' otherwise. 

The algorithm outputs ``$\rho = \sigma$" on all inputs $\rho = \sigma$ with certainty. On input $\rho \in \Cc$, the algorithm $\A$ returns, with probability at least $(1-\delta)$ a description of a state $\pi$ that is $\varepsilon$ close to the input. Our algorithm then uses this information to find a state $\nu \in \Cc$, such that $d_\Tr(\pi, \nu) \leq \varepsilon$. We have from reverse triangle inequality that:
\begin{align}
    |\Tr(\Pi_+ (\rho - \sigma) )| \geq  |\underbrace{|\Tr(\Pi_+ (\sigma - \nu) )|}_{\TRD(\sigma, \nu)} - \underbrace{|\Tr(\Pi_+ (\nu-\rho))|}_{< 2\varepsilon} | \geq \TRD(\nu, \sigma) - 2\varepsilon > 2 \tau,
\end{align}
where we used that $|\Tr(\Pi_+ (\nu-\rho))| \leq \TRD(\nu, \rho) \leq  \TRD(\nu, \pi) + \TRD(\rho, \pi) < 2\varepsilon$.\footnote{Note that this inequality is maximized if $\nu \neq \rho$. This can happen if the input state $\rho \in \Cc$ is less than $2\varepsilon$ far from another state $\rho' \in \Cc$ and the learning algorithm outputs $\pi$ that is closer to $\rho' \in \Cc$. } It follows that: 
\begin{align}
    |R- \Tr(\Pi_+ \sigma)| &\geq |\Tr(\Pi_+(\rho-\sigma))|-\underbrace{|R-\Tr(\Pi_+ \rho)|}_{\leq \tau}  > \tau,
\end{align}
 The algorithm outputs ``$\rho \in \Cc$'' with probability at least $(1-\delta)$, as expected. 
 \end{proof}

 For completeness, we also include a proof of a lower bound on the learning complexity by a decision problem hidden completely (that is, $\sigma \in \Cc$) inside of $\Cc$.
 \begin{lemma}[Learning is as hard as deciding, alternative take]\label{lem:learn_decide_2}
     Let $\D \subset \Cc$ and let $\sigma \in \Cc$, $\sigma \notin \D$, $\TRD(\D, \sigma) > \varepsilon$ and $\varepsilon \geq \tau > 0$. Then: 
     \begin{align}
         \QSQ^{\varepsilon, \delta}_{\tau}(\Cc) \geq \QSD_\tau^\delta(\D, \sigma). 
     \end{align}
 \end{lemma}
 \begin{proof}
    Let $\A$ be a statistical $\varepsilon, \delta$ learning algorithm for $\Cc$ that uses $\Qstat(\tau)$ queries.
    On input $\rho \in \Cc$, the algorithm $\A$ outputs (with probability at least $1-\delta$) a state $\nu$, such that $\TRD(\nu, \rho) \leq \varepsilon$ and uses $\QSQ_\tau^{\varepsilon, \delta}(\Cc)$ many queries. Output ``$\rho = \sigma$'' if $\TRD(\nu, \sigma) < \varepsilon$, otherwise output ``$\rho \in \D$''. The algorithm clearly succeeds with probability at least $1-\delta$.
 \end{proof}

\subsection{Quantum statistical dimension to bound the decision problem}
 
With this lemma, in order to lower bound $\QSQ$ learning it suffices to lower bound $\QSD$, which we do by the quantum statistical dimension that we define now. 

\begin{definition}[Quantum Statistical Dimension]
Let $\tau \in [0,1]$ and $\mu$ be a distribution over a set of $n$-qubit quantum states $\Cc$ and $\sigma \notin \Cc$ be an $n$-qubit  state. Define the maximum covered~fraction:
\begin{align}
    \kappa_\tau\textsf{-\textsf{frac}}(\mu, \sigma) &= \max_{M: \| M \| \leq 1} \lbrace \Pr_{\rho \sim \mu} \left[|\Tr(M(\rho - \sigma))| > \tau \right]\rbrace.
\end{align}
The quantum statistical dimension is: 
\begin{align}
    \QSDA_\tau^\delta(\Cc, \sigma)) = \sup_{\mu} \left[\kappa_\tau\textsf{-\textsf{frac}}(\mu, \sigma) \right]^{-1},
    \label{eq:qsd}
\end{align}
where the supremum is over distibutions over $\Cc$.
\end{definition}
This definition is essentially the same as Feldman's definition of randomized statistical dimension in \cite{Feldman16}, but uses the difference between the expectation values of quantum observables. Sections~4, 5 and 6 of our work show that this has several interesting consquences. The following lemma, following similarly from Feldman's work \cite[Lemma 3.8]{Feldman16}, will be convenient later: 
\begin{lemma}\label{lem:vonneumman}
    Let $\tau > 0$, $\Cc$ be a set of quantum states and $\sigma \notin \Cc$ be another quantum state. Let $d$ be the smallest integer such that there exists a distribution $\nu$ over $\Qstat$ queries $M$ satisfying
    $$
    \forall \hspace{1mm} \rho \in \Cc: \quad \Pr_{M \sim \nu} \left[|\Tr(M(\rho - \sigma))| > \tau \right] \geq 1/d,
    $$
    then $d=\QSDA_\tau(\Cc,\sigma)$.
    \label{Lemma:RandomCover}
\end{lemma}
\begin{proof}
See also \cite[Lemma 3.8.]{Feldman16}, which we generalize here. Suppose that the $\Qstat$ tolerance is fixed to $\tau$ and let $\mathcal{M}$ be the set of all valid $\Qstat$ queries. Define $G: \Cc \times \mathcal{M} \rightarrow \lbrace 0, 1 \rbrace$ as $G(\rho, M) = \delta [|\Tr(M(\rho-\sigma))| > \tau]$, where $\delta [\cdot]$ is the indicator function. Let $\mu$ be a distribution over $\Cc$ and let $\nu$ be a distribution over $\mathcal{M}$. Consider the bilinear function of $\mu, \nu$: \begin{align} F(\mu, \nu) = \int_{\mathcal{M}} d \nu(M) \int_{\Cc} d \mu(\rho)  G(M, \rho)  = \text{Pr}_{\rho \sim \mu} \text{Pr}_{M \sim \nu}[|\Tr(M(\rho-\sigma))| > \tau]. \end{align} 
Note that $\mathcal{M}$ forms a compact subset of $\C^{d \times d }$. We further assume that $\Cc$ is closed and thus also forms a compact subset of $\C^{d \times d}$. Then, the spaces of probability distributions on $\mathcal{M}$ and $\Cc$ form compact, convex spaces with respect to the weak-* topology. It follows by Sion's minimax theorem that:
\begin{align}
    \min_{\mu} \max_{\nu} F(\mu, \nu) = \max_{\nu} \min_{\mu} F(\mu, \nu) =: 1/d,
\end{align}
where the optimization is over possible distributions $\mu$ over $\Cc$ and distributions $\nu$ over $\mathcal{M}$.
For a distribution $\mu$ over $\Cc$, there exists an optimal distinguishing measurement $M \in \mathcal{M}$, from which:
\begin{align}
     \min_{\mu} \max_{\nu} F(\mu, \nu) &= \min_{\mu} \max_{M \in \mathcal{M}} \text{Pr}_{\rho \sim \mu} [|\Tr(M(\rho-\sigma))| > \tau].
\end{align}
Observe that:
\begin{align}
    d = \sup_{\mu} (\max_{M \in \mathcal{M}} \text{Pr}_{\rho \sim \mu} [|\Tr(M(\rho-\sigma))| > \tau])^{-1},
\end{align}
which is the definition of $\QSDA_\tau$ by Eq.~\ref{eq:qsd}. Similarly, we have that:
\begin{align}
    d = ( \max_{\nu}\min_{\mu} F(\mu, \nu))^{-1} = \inf_{\nu} (\min_{\rho \in \Cc} \text{Pr}_{\rho \sim \mu} [|\Tr(M(\rho-\sigma))| > \tau])^{-1}.
\end{align}
For a given distribution $\nu$ over $\mathcal{M}$, it then holds for all $\rho \in \Cc$ that $\Pr_{M \sim \nu}[|\Tr(M(\rho-\sigma))| > \tau] \geq 1/d$. This is the definition in Lemma \ref{Lemma:RandomCover}.
\end{proof}

We now show that the $\QSD$ complexity is lower bounded by $\QSDA$. 
\begin{lemma}
    For every $\sigma\notin\Cc$ and $\tau \in [0,1]$, we have that
    \begin{align} 
    \QSD_\tau^\delta(\Cc,\sigma)\geq (1-2\delta) \QSDA_\tau(\Cc,\sigma).
    \end{align}
\end{lemma}
\begin{proof}
    Let $\A$ be the best algorithm that solves $(\Cc, \sigma)$ with probability at least $1-\delta$ using $q$ $\Qstat(\tau)$ queries $M_1, \ldots M_q$ chosen according to the internal randomness of $\A$ (and based on the received $\alpha_i$). Let $p_\rho = \Pr_\A \left[ \exists i \in [q] | |\Tr(M_i (\rho-\sigma)| > \tau \right]$ be the probability that $\rho \in \Cc$ can be distinguished from $\sigma$ by at least one  of queries. By the correctness of the algorithm, if $\A$ receives $\{\Tr[M_i\sigma]\}_i$ for all qureries, which happens with probability $1-p_\rho$, then $\A$ can output "$\rho\in\Cc$" with probability at most $\delta$. Then, again by the correctness of $\A$, we have that
    \begin{align}
        1-\delta \leq p_\rho + (1-p_\rho)\cdot \delta\ .
    \end{align}
    Since $\delta \geq 1/2$ without loss of generality, this implies that $p_\rho \geq 1-2\delta$ for all $\rho \in \Cc$. Running $\A$ on the responses $\{\Tr[M_i\sigma]\}_i$ and picking one of its queries uniformly randomly then gives $\Pr_{M} \left[ |\Tr(M(\rho-\sigma)| > \tau \right] \geq \frac{1-2\delta}{q}$. Lemma~\ref{Lemma:RandomCover} then implies that $q \geq (1-2\delta)\QSDA_\tau(\Cc,\sigma)$.

\end{proof}

\subsection{Variance and correlation lower bound on quantum statistical dimension}
We now present our main lower bound theorem, wherein we show that there are two combinatorial parameters that can be used to lower bound $\QSDA$, which in turn lower bounds sample complexity in the $\QSQ$ model. Throughout the paper we will use these two two parameters to prove our $\QSQ$ lower bounds. 

\begin{theorem}[Lower bounds]
\label{thm:mainlowerboundonqsq}
Let $\tau > 0$ and $\Cc$ be the class of $n$-qubit states.  Then
\begin{enumerate}
    \item Variance bound: Let $\mu$ be a distribution over $\Cc$, such that $\mathbb{E}_{\rho \sim \mu}[\rho] \notin \Cc$.
    Then: 
\begin{align}
    \QSDA_{\tau}(\Cc,  \Exp_{\rho \sim \mu}[\rho]) &\geq  {\tau^2}\cdot \min_{M, \|M\| \leq 1}  \Big(\mathop{\Var}_{\rho\sim\mu} [\Tr[\rho M]]\Big)^{-1},
\end{align}
where $$
\mathop{\Var}_{\rho \sim \mu}[\Tr(\rho M)]=\mathop{\Exp}_{\rho \sim \mu}[\Tr(\rho M)^2]-\Big(\mathop{\Exp}_{\rho \sim \mu}[\Tr(\rho M)]\Big)^2.
$$
\item Average correlation: For a full-rank quantum state $\sigma \notin \Cc$, define $\hat{\rho} := (\rho \sigma^{-1} - \id)$ and:
\begin{align}
     \gamma(\Cc, \sigma) = \frac{1}{|\Cc|^2}\sum_{\rho_1, \rho_2 \in \Cc} |\Tr(\hat{\rho}_1 \hat{\rho}_2 \sigma)|, \quad \kappa^\gamma_\tau\textsf{-\textsf{frac}}(\Cc_0, \sigma)  := \max_{\Cc' \subseteq \Cc_0} \left\lbrace \frac{|\Cc'|}{|\Cc_0|} : \gamma(\Cc', \sigma) > \tau \right\rbrace.
\end{align} Let $\QAC_\tau(\Cc, \sigma) = \sup_{\Cc_0 \subseteq \Cc} (\kappa^\gamma_\tau\textsf{-\textsf{frac}}(\Cc_0, \sigma))^{-1}$. Then,
$\QSDA_{\tau}(\Cc, \sigma)  \geq \QAC_{\tau^2}(\Cc, \sigma).$
\end{enumerate}
    
\end{theorem}

 \begin{proof}

1. Let $\mu$ be a distribution over $\Cc$ and $M$ a hermitian operator satisfying $\Vert M \Vert \leq 1$ and let $\sigma = \E_{\rho \sim \mu}[\rho]$. By Chebyshev's inequality, we have that
\begin{align}
    \Pr_{\rho \sim \mu} \left[ | \Tr(M(\rho-\sigma) \vert| \geq \tau \right]  \leq {\mathop{\Var}_{\rho \sim \mu}[\Tr(\rho M)]}\cdot {\tau^{-2}}, 
\end{align}
where 
$$
\mathop{\Var}_{\rho \sim \mu}[\Tr(\rho M)]=\mathop{\Exp}_{\rho \sim \mu}[\Tr(\rho M)^2]-\Big(\mathop{\Exp}_{\rho \sim \mu}[\Tr(\rho M)]\Big)^2.
$$

Now let $\nu$ be a distribution over the queries $M$ that are made by a randomized algorithm for the many-one decision problem $(\Cc, \E_{\rho\sim \mu}[\rho])$ and let $d(\nu)$ be the smallest integer such that
$$
    \forall \rho \in \Cc: 1/d(\nu) \leq \Pr_{M\sim\nu}\left[ \vert \Tr[M(\rho-\sigma)] \vert > \tau \right]\ .
$$
We then have that
\begin{align}
    1/d(\nu) & \leq \Pr_{M\sim\nu,\rho\sim\mu}\left[ \vert \Tr[M(\rho-\sigma)] \vert > \tau \right]\ ,\\
    & \leq \E_{M\sim \nu}\left[ {\mathop{\Var}_{\rho \sim \mu}[\Tr(\rho M)]}\cdot {\tau^{-2}} \right]\ ,\\
    & \leq \tau^{-2}\cdot\max_M \left[ {\mathop{\Var}_{\rho \sim \mu}[\Tr(\rho M)]}\right]\ .
\end{align}
Minimizing $d(\nu)$ over all $\nu$ and using lemma~\ref{Lemma:RandomCover} completes the proof.

2.  We now prove the average correlation bound. Let $\Cc'\subseteq \Cc$. Let $\hat{\rho} := (\rho \sigma^{-1} - \id)$ and define:
    \begin{align}
        \gamma(\Cc', \sigma) &= \frac{1}{|\Cc'|^2} \sum_{\rho_i, \rho_j \in \Cc'} |\Tr[\hat{\rho_i}\hat{\rho_j}\sigma]|
    \end{align}
    We will first show that for any such $\Cc'$ and any observable $M, \|M\| \leq 1$, we have that:
    \begin{align}
        \left( \sum_{\rho \in \Cc'} |\Tr(M(\rho-\sigma)| \right)^2 \leq |\Cc'|^2 \gamma(\Cc', \sigma).
    \end{align}
     To that end, observe that:
     \begin{equation}
    \begin{aligned}
        \left( \sum_{\rho \in \Cc'} | \Tr(M(\rho-\sigma)| \right)^2 &=  \left( \sum_{\rho \in \Cc'} | \Tr(M\hat{\rho}\sigma)| \right)^2 = \left[  \Tr\left(\sqrt{\sigma} M \sum_{\rho \in \Cc'} \sign (\Tr(M\hat{\rho}\sigma) \hat{\rho}\sqrt{\sigma}\right)\right]^2 \\
        &\leq \Tr(\sigma M^2) \Tr \left(\left[ \sum_{\rho \in \Cc'} \sign (\Tr(M\hat{\rho}\sigma)\hat{\rho}  \right]^2 \sigma \right),
    \end{aligned}
    \end{equation}
    where the above follows from Cauchy-Schwartz inequality.
    Since $\| M \| \leq 1$, we have that $\Tr(\sigma M^2) \leq 1$ and also that:
    \begin{equation}
    \begin{aligned}
        \Tr \left(\left[\sum_{\rho \in \Cc'} \sign (\Tr(M\hat{\rho}\sigma)\hat{\rho}  \right]^2 \sigma \right) &= \sum_{\rho_1, \rho_2 \in \Cc'} \sign (\Tr(M\hat{\rho_1}\sigma)) \sign (\Tr(M\hat{\rho_2}\sigma))  \Tr \left[\hat{\rho}_1 \hat{\rho}_2 \sigma \right] \\
        &\leq |\Cc'|^2 \gamma(\Cc', \sigma).
    \end{aligned}
    \end{equation}
    
 We show the claim by upper-bounding the $\kappa_\tau\textsf{-\textsf{frac}}(\mu, \sigma)$ for $\mu$ uniform over some subset $\Cc_0 \subseteq \Cc$ by $\kappa_\tau^\gamma\textsf{-\textsf{frac}}(\Cc_0, \sigma)$.  Recall that for a distribution $\mu$ over quantum states, we have that:
     \begin{align}
        \kappa_\tau(\mu, \sigma) &= \max_{M, \|M\|\leq 1} \left\lbrace \Pr_{\rho \sim \mu} \left[ |\Tr(M(\rho-\sigma)| > \tau \right]\right\rbrace
    \end{align}
    For a uniform distribution $\mu_{\Cc_0}$ over $\Cc_0 \subseteq \Cc$, this gives:
    \begin{align}
        \kappa_\tau(\mu_{\Cc_0}, \sigma) &= \max_{M, \|M\|\leq 1}  \frac{1}{|\Cc_0|} \sum_{\rho \in \Cc_0} \delta \left[ |\Tr(M(\rho-\sigma)| > \tau \right],
        \label{Eq:uniform}
    \end{align}
    where $\delta[x] = 1$ if the clause $x$ is true, and $0$ otherwise. From here onwards, fix $M$ to be the operator that maximizes the above expression. Let $\Cc' \subseteq \Cc_0$ to be the largest subset of $\Cc_0$, such that $|\Tr(M(\rho - \sigma))| > \tau$ for all $\rho \in \Cc'$. Then $\sum_{\rho \in \Cc'} |\Tr(M(\rho - \sigma)| > |\Cc'|\tau$. Along with $\sum_{\rho \in \Cc'} \delta[|\Tr(M(\rho - \sigma)| > \tau] = |\Cc'|$ this implies that: 
    \begin{align}
        \sum_{\rho \in \Cc_0} \delta \left[ |\Tr(M(\rho-\sigma)| > \tau \right] \leq \max_{\Cc' \subseteq \Cc_0}  \left[ |\Cc'| \delta\left(\sum_{\rho \in \Cc'} |\Tr(M \rho - \sigma)| > |\Cc'| \tau \right)\right].
    \end{align}
    Combining this with Eq.~\eqref{Eq:uniform}, this gives:
    \begin{align}
        \kappa_\tau(\mu_{\Cc_0}, \sigma) &\leq \max_{C' \subseteq C_0} \left\lbrace \frac{|\Cc'|}{|\Cc_0|} \Bigg\vert \sum_{\rho \in \Cc'} | \Tr(M(\rho-\sigma)| > |\Cc'| \tau \right\rbrace.
    \end{align}
    Using $\left(\sum_{\rho \in \Cc'} | \Tr(M(\rho-\sigma)|\right)^2 \leq |\Cc'|^2 \gamma(\Cc', \sigma)$ implies
    \begin{align}
        \kappa_\tau{-\textsf{frac}}(\mu_{\Cc_0}, \sigma) \leq \min_{\Cc' \subseteq \Cc} \kappa_{\tau^2}^\gamma{-\textsf{frac}}(\Cc', \sigma).
    \end{align}
Hence we have that
    \begin{align*}
        \QSDA_\tau(\Cc, \sigma) &= \sup_{\mu}(\kappa_\tau(\mu, \sigma)^{-1}) \geq (\kappa_\tau(\mu_{\Cc_0}, \sigma)^{-1}) \geq 
        \max_{\Cc' \subseteq \Cc}(\kappa_{\tau^2}^\gamma{-\textsf{frac}}(\Cc', \sigma)^{-1}) = \QAC_{\tau^2}(\Cc, \sigma).
    \end{align*}
    This proves the lower bounds in the theorem statement.
\end{proof}

In many of the bounds to be proved in the following sections we consider converting a learning problem to a decision problem $\Cc$ versus $\sigma$ where $\min_{\rho \in \Cc}\TRD(\rho,\sigma) \geq \zeta$, where $\zeta$ is some constant. For large enough $\tau$ Lemma~\ref{Lemma:LearningToDeciding} may no longer hold. Fixing an approximation error $\eps$, Lemma~\ref{Lemma:LearningToDeciding} then holds if $\zeta - 2\eps > 2\tau$. Note that the left hand side is some constant and we implicitly assume this upper bound on $\tau$ in the following proofs. This is without loss of generality: the existence of a $\QSQ$ algorithm with tolerance $\tau$ greater than or equal to $\zeta - 2\eps$ then further implies that one exists for all tolerances of smaller value. That is, smaller tolerance cannot increase the query/time complexity. As many of the results are asymptotic, requiring that $\tau$ be at most some constant does not change the results even when $\tau$ appears in the lower bound.


\section{Separations between statistical and entangled measurements}
\label{sec:deg2sep}
In this section we prove our main theorem separating noisy entangled $\QPAC$ learning and $\QSQ$ learning, and next show that for a ``small" circuit one can witness such an exponential separation.
\subsection{Separation between $\QSQ$ and $\QPAC$ with classification noise}
In this section we prove our main theorem. Consider the class of function states 
    $$
    \Cc=\Big\{\ket{\psi_A}=\frac{1}{\sqrt{2^n}}\sum_{x\in \01^n} \ket{x,x^\top Ax \text{ (mod 2} )}:A\in \mathbb{F}_2^{n\times n}\Big\}.
    $$ 
    The sample complexity of learning this class in the following models is given as follows
    \begin{enumerate}
        \item  Entangled measurements: $\Theta(n)$
        \item  Separable measurements: $\Theta(n^2)$
        \item Statistical query learning: $\Omega(\tau^2 \cdot 2^{n/2})$ making $\Qstat(\tau)$ queries.
         \item Entangled $\eta$-random classification noise: $O(\frac{n}{(1-2\eta)^2})$. Algorithm runs in time $O(n^3/(1-2\eta)^2)$.
    \end{enumerate}

Points $(1),(2)$ above were proved in~\cite{arunachalam2022optimal} and we do not prove it here. In the following two theorems we prove points $(3),(4)$ above. We remark that this result partially resolves an open question in~\cite[Question~18]{anshu23}, who asked if separable $\eta$-random classification noise and $\QSQ$ complexity can be separated in sample complexity.

\begin{theorem}\label{thm:mainsep}
 The concept class
    $$
    \Cc=\Big\{\ket{\psi_A}=\frac{1}{\sqrt{2^n}}\sum_{x\in \01^n} \ket{x,x^\top Ax \text{ (mod 2} )}:A\in \mathbb{F}_2^{n\times n}\Big\}
    $$ 
    requires $2^{\Omega(n)}$ many $\Qstat$ queries of tolerance $\tau=1/\poly(n)$ to learn below trace distance $0.05$ with high probability.
\end{theorem}

\begin{proof}
    We prove the hardness for algorithms using $\Qstat$ queries using the variance lower bounding technique in Theorem~\ref{thm:mainlowerboundonqsq}. In particular, we show an exponentially small upper bound for the variance for any observable: for every $n+1$ qubit operator $M$ such that $\Vert M \Vert \leq 1$ we have~that
        \begin{align}
        \label{eq:variancebound}
            \Var_{A}\left(\Tr[M\psi_A]\right) = 2^{-\Omega(n)},
        \end{align}
        where we let $\psi_A=\ketbra{\psi_A}{\psi_A}$ for notational simplicity. To apply our results linking learning and decision problems we note that
        \begin{align*}
            \TRD(\psi_A, \E_B[\rho_B]) & \geq 1-\sqrt{\E_B[\vert \braket{\psi_A}{\psi_B}]\vert^2} \geq 1-\sqrt{\frac{(2^{n(n+1)/2}-1)\cdot 9/16 + 1}{2^{n(n+1)/2}}} \geq 1 - \sqrt{17/32}\ ,
        \end{align*}
        where the first inequality follows from the lower bound on trace distance by fidelity \cite{Fuchs2006Cryptographic} and the second by Fact~\ref{fact:zippel} and that $\braket{\psi_A}{\psi_B} = \Pr_{x}[f_A(x) = f_B(x)]$. Fix $\eps = 0.05$. Then Lemma~\ref{Lemma:LearningToDeciding} holds if $\tau < 0.085$, which we assume without loss of generality as previously discussed\footnote{The choice of $\eps = 0.05$ is arbitrarily and done for readability. A similar result holds for any $\eps < 1/2(1-\sqrt{17/32})$ by the same argument.} Along with Theorem~\ref{thm:mainlowerboundonqsq}, we obtain our lower bound on the $\QSQ$ complexity of learning $\Cc$.

        It remains to establish Eq.~\eqref{eq:variancebound}. To this end, we need to understand
    \begin{align}
        \text{Var}_A(\Tr[M\psi_A]) & = \E_A[\Tr[M\psi_A]^2] - (\E_A[\Tr[M\psi_A]])^2
    \end{align}
    To do so, we decompose $\psi_A$ as follows. For every $f_A: \{0,1\}^n \rightarrow \{0,1\}$ given by $f_A(x) = x^\top A x$ let $\ket{\psi_A} = \frac{1}{\sqrt{2^n}}\sum_x \ket{x,f_A(x)}$ and $\ket{\phi_A} = \sum_x (-1)^{f_A(x)} \ket{x}$. For convenience we let $\ket{u} = \frac{1}{\sqrt{2^n}} \sum_x \ket{x}$. Then we see that
    \begin{align*}
        (\id \otimes H)\psi_A (\id \otimes H) & = \frac{1}{2}\sum_{x,y,a,b} (-1)^{a\cdot f_A(x)+b\cdot f_A(y)}\ketbra{x,a}{y,b}\\
        & = \frac{1}{2}\Big(\ketbra{\phi_A}{\phi_A}\otimes \ketbra{1}{1}-\frac{1}{2}\ketbra{\phi_A}{u}\otimes \ketbra{1}{0}-\frac{1}{2}\ketbra{u}{\phi_A}\otimes \ketbra{0}{1}+\frac{1}{2}\ketbra{u}{u}\otimes \ketbra{0}{0}\Big)
    \end{align*}
    hence we have that
\begin{align}
\ketbra{\psi_A}{\psi_A}=\frac{1}{2}\Big(\underbrace{\ketbra{\phi_A}{\phi_A}\otimes \ketbra{-}{-}}_{\rho^A_1}-\underbrace{\ketbra{\phi_A}{u}\otimes \ketbra{-}{+}}_{\rho^A_2}-\underbrace{\ketbra{u}{\phi_A}\otimes \ketbra{+}{-}}_{\rho^A_3}+\underbrace{\ketbra{u}{u}\otimes \ketbra{+}{+}}_{\rho^A_4}\Big).
\end{align}
Any $n+1$ qubit observable $M$ can be decomposed as $M = \sum_{a,b} M_{a,b}\otimes \ketbra{a}{b}$ where $a,b \in \{+,-\}$. Since $\Vert M \Vert \leq 1$ we also have that $\Vert M_{a,b} \Vert \leq 1$, however the off-diagonal blocks now no longer need be Hermitian. Instead, $M_{+,-} = M_{-,+}^\dagger$. In an abuse of notation we now discard the last qubit of $\rho^A_i$ and denote the resulting state also as $\rho^A_i$. For convenience we further introduce the notation $M_1 = M_{-,-}$, $M_2 = M_{+,-}$, $M_3 = M_{-,+}$, and $M_4 = M_{+,+}$ Thus, we see that $\Tr[M\psi_A] = \frac{1}{2}\sum_i\Tr[M_i\rho^A_i]$ and further the variance can be written as 
\begin{align}
    \text{Var}_A(\Tr[M\psi_A]) & = \frac{1}{4}\Big(\E_A\Big[\sum_{i,j} \Tr\left[M_i\rho^A_i\right]\Tr\left[M_j\rho^A_j\right]\Big] - \Big(\E_A\Big[\sum_i\Tr\left[M_i\rho^A_i\right]\Big]\Big)^2\Big)\\
    & = \frac{1}{4}\sum_{i,j}\Big(\Exp_A\left[\Tr\big[M_i \rho^A_i\big]\cdot \Tr\big[M_j \rho^A_j\big]\right]- \Tr\left[ M_i\Exp_A\left[\rho^A_i\right]\right]\cdot\Tr\left[M_j \Exp_A\left[\rho^A_j\right]\right]\Big)\\
    & = \frac{1}{4}\sum_{i,j}\Big( \Tr\left[ (M_i\otimes M_j)\E_A\left[ \rho_i^A \otimes \rho_j^A\right]\right] -\Tr\left[ M_i\Exp_A\left[\rho^A_i\right]\right]\cdot\Tr\left[M_j \Exp_A\left[\rho^A_j\right]\right]\Big).
\end{align}
Below we drop the factor of $1/4$ and bound the magnitude of each term for $i,j \in [4]$. We show that each term is exponentially small. To do so, we use the following fact which will be proven later.
\begin{fact}
\label{fact:phasestatedecomp}
 We have the following
 \begin{align*}
 &\Exp_A[\ket{\phi_A}]=\ket{0^n}/\sqrt{2^n},\qquad \Exp_A[\ket{\phi_A}^{\otimes 2}]=\ket{\Phi^+}/\sqrt{2^n},\qquad \Exp_A[\ketbra{\phi_A}{\phi_A}]=\id/2^n,\\
 &  \Exp_A[\ket{\phi_A}\otimes \bra{\phi_A}] = \frac{1}{2^n}\sum_x \ket{x}\otimes \bra{x},\quad \Exp_A[\ket{\phi_A}\otimes \ketbra{\phi_A}{\phi_A}] = \frac{1}{2^{3n/2}}\ket{0}\otimes\ketbra{0}{0}, \\
 &\Exp_A[\ketbra{\phi_A}{\phi_A}\otimes \ket{\phi_A}] = \frac{1}{2^{3n/2}}\Big(\sum_x \ketbra{x}{x}\otimes \ket{0} + \ketbra{x}{0}\otimes \ket{x} + \ketbra{0}{x}\otimes\ket{x} - 2\ketbra{0}{0}\otimes \ket{0}\Big),\\
 &\Exp_A [\ketbra{\phi_A}{\phi_A}^{\otimes 2}] = 
\frac1{4^n} (\id+\SWAP) + \frac{1}{2^n} \ketbra{\Phi^+}{ \Phi^+} - \frac{2}{4^n} \sum_x \ketbra{x,x}{x,x},
\end{align*}
where $\SWAP$ swaps two $n$-qubit registers, i.e., $\SWAP \ket{\psi}\otimes \ket{\phi} = \ket{\phi}\otimes\ket{\psi}$ and
$\ket{\Phi^+} = 2^{-n/2} \sum_x \ket{x,x}$.
is the EPR state of $2n$ qubits.
\end{fact}
First note that $\rho^A_4$ does not depend on $A$. Thus, for all $i \in [4]$ we have that 
\begin{align}
    \E_A\left[\Tr\left[(M_i\otimes M_4)(\rho_i^A \otimes \rho_4^A)\right]\right] & = \Tr\left[M_i \E_A[\rho_i^A]\right]\cdot\Tr\left[M_4\E_A[\rho_4^A]\right]\ .
\end{align}
The contribution to the variance from these cases equals $0$. We are left with analyzing $i,j\in [3]$ and do so separately below. 

\textbf{Case $i=j=1$.} Using Fact~\ref{fact:phasestatedecomp} above, we get that 
$$
\Tr\big[ M_1 \Exp_A[\rho^A_1]\big]=\Tr\big[M_1 \cdot \Exp_A\left[\ketbra{\phi_A}{\phi_A}\right]\big]=\Tr[M_1]/2^n.
$$
Next, observe that
\begin{align}
    \Exp_A \big[\Tr(M_1\rho^A_1)^2\big]&=\Tr\left[M_1\otimes M_1\cdot \Exp_A[\rho^A_1\otimes \rho^A_1]\right]\\
    &=\Tr\Big[M_1\otimes M_1\cdot \Big(\frac1{4^n} (\id+\SWAP) + \frac{1}{2^n} \ketbra{\Phi^+}{ \Phi^+} - \frac{2}{4^n} \sum_x \ketbra{x,x}{x,x}\Big)\Big]\\
    &=\frac1{4^n}\big(\Tr [M_1]\big)^2 +\frac1{4^n}\Tr\left[M_1^2\right] + \frac{1}{2^n} \bra{\Phi^+} M_1 \ket{\Phi^+} - \frac{2}{4^n} \sum_x M_1(x,x)^2\\
    &\leq\frac1{4^n}\left(\Tr [M_1]\right)^2 +\frac1{2^n} + \frac{1}{2^n},
\end{align}
where the third equality used that $\Tr[M_1\otimes M_1\cdot \SWAP]=\Tr\left[M_1^2\right]$, the fourth equality used that $\Tr\left[M_1^2\right]\leq 2^n$ and $\Vert M_1^{\otimes 2}\Vert \leq 1 $. Implicitly we have used that $M_1$ is Hermitian and $\Vert M_1 \Vert \leq 1 $. Hence we have that variance term contribution is
\begin{align*}
&\Exp_A \left[\Tr\big[M_1 \rho^A_1\big]\cdot \Tr\big[M_1 \rho^A_1\big]\right]- \Tr\big[ M_1\Exp_A[\rho^A_1]\big]\cdot\Tr\big[ M_1 \Exp_A[\rho^A_1]\big]\\
&\leq \frac1{4^n}\big(\Tr [M_1]\big)^2 +\frac2{2^n}-\Big(\frac{\Tr[M_1]}{2^n}\Big)^2= 2/2^n.
\end{align*}
Since this term must be non-negative, the norm is bounded by $2/2^n$ as well.

\textbf{Case $i=j=2$.} Using Fact~\ref{fact:phasestatedecomp} above, we get that 
$$
\Tr\left[ M_2 \Exp_A[\rho^A_2] \right]=\Tr\left[ M_2 \Exp_A[\ketbra{\phi_A}{u}]\right]=\langle 0|M_2| u\rangle/\sqrt{2^n}.
$$
Next note that
\begin{align}
    \E_A\left[\Tr[M_2\rho_2^A]^2\right] & = \Tr\left[(M_2 \otimes M_2) (\E_A[\rho_2^A \otimes \rho_2^A])\right]\\
    & = \Tr\left[(M_2 \otimes M_2)  (\E_A[\ket{\phi_A}^{\otimes 2}]\bra{u}^{\otimes 2})\right]\\
    & = \frac{1}{\sqrt{2^n}}\bra{u,u }M_2 \otimes M_2\ket{\Phi^+} = \frac{1}{2^n}\sum_x \bra{u} M_2 \ket{x}^2 = \frac{1}{4^n}\sum_x \left(\sum_y M_2(y,x)\right)^2
\end{align}
We now bound the norm of each of these terms individual as then, by triangle inequality, the norm of the contribution from the case is exponentially small as well.
\begin{align}
    \left\vert (\langle 0|M_2| u\rangle/\sqrt{2^n}) \right\vert & \leq \frac{1}{\sqrt{2^n}} \sqrt{\Vert \ket{0} \Vert^2 \ \Vert M_2 \ket{u} \Vert^2} \leq \frac{1}{\sqrt{2^n}}\ ,
\end{align}
and thus $\left\vert \Tr\left[M_2 \Exp_A[\rho^A_2]\right]^2 \right\vert \leq \frac{1}{2^n}$. For the other term we use that 
\begin{align}
    \left\vert \sum_x \left(\sum_y M_2(y,x)\right)^2 \right\vert  \leq \sum_x \left\vert \sum_y M_2(y,x) \right\vert^2.
\end{align}
Then we can rewrite this as $\frac{1}{2^n}\Vert M_2 \ket{u} \Vert_2^2 \leq \frac{1}{2^n} \Vert M_2 \Vert^2 \leq \frac{1}{2^n}$. Thus, the norm of the contribution from this case is upper bounded by $2/2^n$.

 \textbf{Case $i=j=3$.}  Since $M_3 = M_2^\dagger$ and $\rho_3^A = (\rho_2^A)^\dagger$, this is the same as the case $i=j=2$, thus the norm of this case is upper bounded by $\frac{2}{2^n}$ as well.

\textbf{Case $i=2$ and $j=3$} Using Fact~\ref{fact:phasestatedecomp} note that
\begin{align}
    \Tr[(M_2 \otimes M_3) \E_f[\rho_2^A \otimes \rho_3^A]] & = \Tr[(M_2\otimes M_3) (\id \otimes \ket{u})\E_A[\ket{\phi_A}\otimes \bra{\phi_A}] (\bra{u} \otimes \id)]\\
    & = \frac{1}{2^n}\sum_x \Tr[M_2 \otimes M_3 (\ketbra{x}{u}\otimes \ketbra{u}{x})
\end{align}
 Now we use that $M_3 = M_2^\dagger$ to rewrite this as $\frac{1}{2^n}\sum_x \vert \bra{u}M_2\ket{x}\vert^2$
 \begin{align}
    \frac{1}{4^n} \sum_x \vert \bra{u}M_2\ket{x}\vert^2 & = \frac{1}{2^n} \Vert M_2 \ket{u} \Vert_2^2  \leq \frac{1}{2^n}
 \end{align}
 We have already shown the subtracted terms to be exponentially small in magnitude and thus the magnitude of this case must be upper bounded by $2/2^n$ as well.
 
 \textbf{Case $i=1$ and $j=2,3$.} Here we work out $j=2$ as the result then holds similarly for $j=3$. 
 \begin{align}
    \left\vert \Tr[(M_1 \otimes M_2) \E_A[\rho_1^A \otimes \rho_2^A]] \right\vert & = \left\vert \Tr[(M_1 \otimes M_2) \E_A[\ketbra{\phi_A}{\phi_A}\otimes \ket{\phi_A}](\id \otimes \bra{u})]\right\vert \\
    \quad = \Big\vert \frac{1}{2^{3n/2}}\Tr\Big[(M_1 \otimes M_2) & \Big(\sum_x \ketbra{x}{x}\otimes \ketbra{0}{u}+ \ketbra{x}{0}\otimes\ketbra{x}{u}+ \ketbra{0}{x}\otimes\ketbra{x}{u} - 2\ketbra{0}{0}\otimes\ketbra{0}{u}\Big)\Big] \Big\vert
 \end{align}
We deal with each term in the trace above separately. First:
\begin{align}
    \frac{1}{2^{3n/2}}\left\vert \Tr\left[M_1\otimes M_2 \left(\sum_x \ketbra{x}{x} \otimes \ketbra{0}{u}\right)\right]\right| & = \frac{1}{2^{3n/2}}\vert \Tr[M_1]\cdot \langle u \vert M_2 \vert 0 \rangle \vert \leq \frac{1}{2^{n/2}}\ .
\end{align}
The next two terms are similar and we bound the first here (which implies the same upper bound for the second using nearly identical steps).
\begin{align}
    \frac{1}{2^{3n/2}}\left\vert \Tr[M_1\otimes M_2 \left(\sum_x \ketbra{x}{0}\otimes \ketbra{x}{u}\right)\right] & = \frac{1}{2^n}\vert \langle 0,u \vert M_1 \otimes M_2 \vert \Phi^+ \rangle\vert \leq \frac{1}{2^n} \ .
\end{align}
The last remaining term is bounded as follows:
\begin{align}
    \frac{1}{2^{3n/2}}\left\vert \Tr[M_1 \otimes M_2 \ketbra{0}{0}\otimes \ketbra{0}{u}]\right\vert & =  \frac{1}{2^{3n/2}} \left\vert \Tr[M_1 \ketbra{0}{0}]\right\vert \cdot \left\vert \Tr[M_2 \ketbra{0}{u}] \right\vert  \leq \frac{1}{2^{3n/2}}\ .
\end{align}
Thus, the contribution from the second moments is of magnitude at most $O(2^{-n/2})$. While $\left\vert \Tr[M_1 \E_A[\rho_1^A]] \right\vert$ may be large (up to $1$), we also have that $\left\vert \Tr[M_2 \E_A[\rho_2^A]] \right\vert \leq \frac{1}{\sqrt{2^n}}$ and thus all terms in the case are exponentially small in norm as well. We have thus shown that the norms of the contributions for each case are all exponentially small. Thus, the variance must be exponentially small as well. It remains to prove Fact~\ref{fact:phasestatedecomp} which we do now.

\begin{proof}[Proof of Fact~\ref{fact:phasestatedecomp}]
Most of the desired expectation values stem from the following observation: for the uniform distribution over upper triangular matrices $A$, we have that
    \begin{align} 
    \E_A[(-1)^{x^\top A x + y^\top A y + z^\top A z}] & = \E_A[(-1)^{\langle X + Y + Z, A\rangle}]\\
    & = \delta_{X+Y+Z,\textbf{0}} = \delta_{x,y}\delta_{z,0} + \delta_{x,z}\delta_{y,0} + \delta_{y,z}\delta_{x,0} - 2\delta_{x,0}\delta_{y,0}\delta_{z,0}
    \end{align}\,
where $X$, $Y$, and $Z$ are defined as $xx^T$, $yy^T$, and $zz^T$ respectively. The second equality follows from $E_z[(-1)^{\langle x,z \rangle}]=\delta_{x,0}$ and the third from the fact that $X+Y$ is rank $2$ over $\mathbb{F}_2$ unless $x=y$. Now the first three equalities are now easy to see: set $y=z=0$ and
$$
    \E_A[\ket{\phi_A}] = \E_A\left[\frac{1}{\sqrt{2^n}}\sum_x (-1)^{x^\top A x}\ket{x}\right]=\frac{1}{\sqrt{2^n}}\sum_x \E_A\left[(-1)^{x^\top A x}\right]\ket{x} =\ket{0^n}/\sqrt{2^n}\ .
$$
Similarly, setting $z=0$ yields
$$
    \E_A\left[\ket{\phi_A}^{\otimes 2}\right]=\E_A\left[\frac{1}{2^n}\sum_{x,y} (-1)^{x^\top A x+y^\top Ay}\ket{x,y}\right]=\frac{1}{2^n}\sum_{x,y} \E_A[(-1)^{x^\top A x+y^\top Ay}]\ket{x,y}\\
    =\frac{1}{2^n}\sum_{x}\ket{x,x}\ .
$$
Similar reasoning implies $\E_A[\ketbra{\phi_A}{\phi_A}]=\id/2^n$, $\E_A[\ket{\phi_A}\otimes \bra{\phi_A}] = \frac{1}{2^n}\sum_x \ket{x} \otimes \bra{x}$, and
$$
    \E_A[\ketbra{\phi_A}{\phi_A}\otimes \ket{\phi_A}] = \frac{1}{2^{3n/2}}\left(\sum_x \ketbra{x}{x}\otimes \ket{0} + \ketbra{x}{0}\otimes \ket{x} + \ketbra{0}{x}\otimes\ket{x} - 2\ketbra{0}{0}\otimes \ket{0}\right)\ .
$$

The final decomposition of $\E_A[\ketbra{\phi_A}{\phi_A}^{\otimes 2}]$ follows from~\cite[Proposition~2]{arunachalam2022optimal}.
    \end{proof}

The proof of the fact concludes the proof of the theorem.
\end{proof}

\begin{theorem}
\label{thm:noisypacupperbound}
 The concept class
    $$
    \Cc=\Big\{\ket{\psi_A}=\frac{1}{\sqrt{2^n}}\sum_{x\in \01^n} \ket{x,x^\top Ax \text{ (mod 2} )}:A\in \mathbb{F}_2^{n\times n}\Big\}
    $$ 
can be learned in the $\eta$-random classification model, using $O(n/(1-2\eta)^2)$ copies of the noisy state and time $O(n^3/(1-2\eta)^2)$.
\end{theorem}

\begin{proof}
    Below, let $\ket{\phi_f}=\frac{1}{\sqrt{2^n}}\sum_x (-1)^{f(x)}\ket{x}$. In the classification noise model, we are given copies~of
$$
\ket{\psi_f}=\frac{1}{\sqrt{2^n}}\sum_x \ket{x} \otimes \big(\sqrt{1-\eta}\ket{f(x)}+\sqrt{\eta}\ket{\overline{f(x)}}\big).
$$
We first show that we can convert $\ket{\psi_n}$ to $\ket{\phi_f}$. In order to do so, observe the following
\begin{align*}
    (\id \otimes H)\ket{\psi_f}&=\frac{1}{\sqrt{2^{n+1}}}\sum_{x,b}\ket{x} \big( (-1)^{b\cdot f(x)}\ket{b}+\sqrt{\eta}(-1)^{b\cdot \overline{f(x)}}\ket{b}\big)\\
    &=\frac{1}{\sqrt{2^{n+1}}}\sum_{x,b} \big(\sqrt{1-\eta} (-1)^{b\cdot f(x)}+\sqrt{\eta}(-1)^{b\cdot \overline{f(x)}}\big) \ket{x,b}.
\end{align*}
Now, measuring the last qubit, the probability of seeing $b=1$ is given by 
\begin{align*}
    &\Big\| \frac{1}{\sqrt{2^{n+1}}}\sum_{x} \big(\sqrt{1-\eta} (-1)^{f(x)}-\sqrt{\eta}(-1)^{f(x)}\big) \ket{x}\Big\|_2^2\\
    &=\frac{1}{2^{n+1}}\sum_x \big(\sqrt{1-\eta} (-1)^{f(x)}-\sqrt{\eta}(-1)^{f(x)}\big)^2=\frac{1}{2} (\sqrt{1-\eta}-\sqrt{\eta})^2:=p,
\end{align*}
and the post-measurement state is given by $\ket{\phi_f}$. Hence with probability exactly $p=\frac{1}{2}(1-2\sqrt{\eta(1-\eta)})\geq \frac{1}{4}(1-2\eta)^2$, we can convert $\ket{\psi_f}$ to $\ket{\phi_f}$. 

Now we focus on the concept class $\{f_A(x)=x^\top A x\}_A$. The learning algorithm first takes $O(1/(1-2\eta)^2)$ copies of $\ket{\psi_A}$ to produce two copies of $\ket{\phi_A}$. Note that the algorithm knows when it succeeded, i.e., when the measurement of the last qubit is $1$, the algorithm knows that the above procedure performed the transformation $\ket{\psi_A}^{\otimes 2}\rightarrow \ket{\phi_A}^{\otimes 2}$. Now using Fact~\ref{fact:deg2} we can learn $f_A$ given $O(n)$ copies of $\ket{\phi_A}$ and $O(n^3)$ time.  Overall, the sample complexity and time complexity of the procedure is $O(n/(1-2\eta)^2)$ and $O(n^3/(1-2\eta)^2)$ respectively.
\end{proof}

\subsection{Smallest circuit class witnessing separation}
In the previous section we saw that the concept class of quadratic functions separated $\QPAC$ from $\QSQ$. Observe that states in this concept class can be prepared by circuits of size $O(n^2)$ and depth $O(n)$ consisting of $\{\textsf{Had},\textsf{X},\textsf{CX}\}$ gates. A natural question is, can states prepared by \emph{smaller} circuits also witness such a separation between $\QPAC$ and $\QSQ$? Below we answer this in the positive, by using a simple padding argument inspired by a prior work of Hinshe et al.~\cite{hinsche2022single}.

\begin{theorem}\label{thm:smallckt}
    Let $\alpha \in (0,1)$ there exists a family of $n$ qubit Clifford circuits of depth $d=(\log n)^{1/\alpha}$ and size $d^2$ that requires $2^{\Omega(d)}$ $\Qstat$ queries to learn the state to error $\leq 0.05$ in trace~distance.
\end{theorem}
\begin{proof}
The idea is to ``pad" a family of circuits with auxilliary qubits. In the previous section, from Theorem~\ref{thm:mainsep} we saw that the set of example states $\{\ket{\psi_A}=\frac{1}{\sqrt{2^n}}\sum_x \ket{x,x^\top A x}\}_{A}$, is hard to learn to trace distance $0.05$. Instead of the example state $\ket{\psi_A}$ now instead consider the ``padded state" $\ket{\psi_A}\otimes \ket{0}^{k(n)}$. Say a $\QSQ$ algorithm learns these padded states with the set of $\Qstat$ queries given by $\{M_i\}_i$ (which are random variables). Let us decompose each $M_i$ as $M_i = \sum_{x,y\in \01^{k(n)}}M_i^{x,y}\otimes \ketbra{x}{y}$, where $\Vert M_i^{x,y}\Vert \leq 1$ and $M_i^{x,x}$ is Hermitian. Since the auxiliary qubits are fixed, it is clear that 
$$
\Tr\left[M \cdot \ketbra{\phi_A}{\phi_A} \otimes \ketbra{0}{0}^{\otimes k(n)}\right] = \Tr\left[M_i^{0,0} \phi_A\right].
$$
Furthermore, we can assume without loss of generality that the algorithm always outputs a state of the form $\pi \otimes (\ket{0}\bra{0})^{\otimes k(n)}$ (as otherwise we could improve $\A$ by requiring it to do so). Thus, a $\QSQ$ algorithm for the padded states up to trace distance $0.05$ implies a $\QSQ$ algorithm with queries $\left\{M_i^{0,0}\right\}_i$ for learning $\Cc =\{\ket{\psi_A}\}_A$ up to trace distance $0.05$. Say that this algorithm uses at most $t$ $\Qstat$ queries. Then Theorem~\ref{thm:mainsep} implies that $t \geq 2^{\Omega(n)}$. The state is now composed of $m = k(n)+n$ qubits. Pick $k=2^{n^\alpha}$ for some $\alpha<1$, so $m=\Theta(k(n))$ and $n = (\log k)^{1/\alpha}$. Then we have that $t \geq 2^{\Omega(\log m)^{1/\alpha})}$. To conclude the theorem, note that $\ket{\psi_A}$ has a circuit of size $O(n^2)$ and depth $O(n)$. Thus, the padded states can be prepared with circuits of size $O(\log m)^{2/\alpha}$ and depth $O(\log m)^{1/\alpha}$.
\end{proof}

\section{New upper and lower bounds on $\QSQ$ learning states}
\label{sec:examples}
In this section we first give a couple of classes of states which can be learned in the $\QSQ$ framework before discussing lower bounds for other class of states.
\subsection{New upper bounds}
We first prove that the class of functions that are $k$-Fourier-sparse Boolean functions on $n$ bits, i.e.,
$$
\Cc_1=\{f:\01^n\rightarrow \01: |\textsf{supp}(\widehat{f})|=k\}
$$
 can be learned in time $\poly(n,k)$ in the $\QSQ$ model. This generalizes the results in~\cite{arunachalam2021two,arunachalam2020quantum}, which showed that showed parities and $O(\log n)$-juntas (which are a subset of Fourier-sparse functions)  are $\poly(n)$-time learnable.\footnote{We remark that the same proof also shows that $k$-term DNF formulas are learnable: for every $g\in \Cc_1$, there exists $S$ s.t. $|\widehat{g}(S)|\geq 1/k$ and the proof of Theorem~\ref{thm:learningsparsefunctions} can identify such an $S$  using $\QSQ$ queries and then one can use the algorithm of Feldman~\cite{Feldman12:dnfwithoutboosting} for learning  the unknown DNF formulas.}  We  observe that the quantum coupon collector problem, i.e., learnability of
 $$
 \Cc_2=\{S\subseteq [n]: |S|=k\}
 $$
 considered in~\cite{arunachalam2020quantumcoupon} can be implemented in $\QSQ$. Finally, we also observe that one can learn codeword states defined in~\cite{arunachalam2018optimal}: consider an $[n,k,d]_2$ linear code $\{Mx: x\in \01^k\}$ where $G \in \mathbb{F}_2^{n\times k}$ is a rank-$k$ generator matrix of the code, $k=\Omega(n)$, and distinct codewords have Hamming distance at least $d$, then define the concept class 
 $$
 \Cc_3=\{f_x(i)=(Gx)_i:x\in \01^k\},
 $$
 where $G$ is known the learning algorithm. Below we show we can learn $\Cc_3$ in the $\QSQ$ model. Prior learning protocols~\cite{arunachalam2018optimal,arunachalam2020quantumcoupon,arunachalam2021two,arunachalam2020quantum}  showed that these concept classes are learnable with quantum examples (a stronger model than $\QSQ$) whereas here we show they are learnable in the weaker $\QSQ$ framework. Before we prove this, we will use the following~lemmas.
 	\begin{lemma}{\cite[Theorem~12]{gopalan:fouriersparsity}}
		\label{lemma:granularityofsparse}
		Let $k\geq 2$. The Fourier coefficients of a $k$-Fourier-sparse Boolean function $f:\01^n\rightarrow \pmset{}$ are integer multiples of $2^{1-\lfloor \log k\rfloor}$.
	\end{lemma}

	\begin{lemma}{\cite[Theorem~4.4]{arunachalam2020quantum}}
		\label{thm:goldreichlevin} Let $f:\pmset{n}\rightarrow \pmset{}$,  $\tau \in (0,1]$. There exists a $\poly(n,1/\tau,\ell)$-time quantum statistical learning algorithm that with high probability outputs $U=\{T_1,\ldots,T_\ell\}\subseteq [n]$ such that: (i) if $|\widehat{f}(T)|\geq \tau$, then $T\in U$; and (ii) if $T\in U$, then $|\widehat{f}(T)|\geq \tau/2$.	
	\end{lemma}

\begin{theorem}
\label{thm:learningsparsefunctions}
The concept classes  $\Cc_1,\Cc_2,\Cc_3$ defined above can be learned in the $\QSQ$ model.
\end{theorem}
\begin{proof} 
We first give a learning algorithm for $\Cc_1$. For every $f\in \Cc_1$, observe that it's Fourier coefficients satisfy $|\widehat{f}(S)|\geq 1/k$ by Lemma~\ref{lemma:granularityofsparse}. We can now use Lemma~\ref{thm:goldreichlevin} to collect \emph{all} the non-zero  Fourier coefficients in time $\poly(n,1/\tau,k)$ in the $\QSQ$ model. Call these non-zero coefficients $S_1,\ldots,S_k$. Next, we learn all these Fourier coefficients up to error $\varepsilon/k$ using $\Stat$ queries: for $i\in [k]$, let $\phi(x,b)= b\cdot (-1)^{S_i\cdot x}$ for all $x\in \01^n,b\in \01$, hence $\Exp_x [\phi(x,f(x))]=\Exp_x \big[f(x)\cdot (-1)^{S_i\cdot x}\big]=\widehat{f}(S_i)$. Overall this takes time $O(k)$. Once we obtain all these approximations $\{\alpha_i\}_{i\in [k]}$, we output the function $
		g(x)=\textsf{sign} \Big(\sum_{i\in [k]} \alpha_i \cdot  \chi_{S_i}(x)\Big)$ for every  $x\in \01^{n}.$
Using the same reasoning as in~\cite[Eq.~(7)]{arunachalam2020quantum} it is not hard to see that $g$ is $\varepsilon$-close to $f$ (i.e., $\Pr_x [g(x)=f(x)]\geq 1-\varepsilon$).

We next give a learning algorithm for $\Cc_2$. Let $S \subseteq [n]$ of size $k$. Given copies of $\frac{1}{\sqrt{k}} \sum_{i \in S} \ket{i}$, learn $S$. We now show how to learn $S$ in $\QSQ$  using $k \log n$ $\Qstat$ queries. Let $M_1=\sum_{i=1}^{n/2} \ketbra{i}{i}$. This satisfies $\|M_1\|\leq 1$ and $M_1$ can be implemented using a $\poly(n)$-sized circuit. Observe that 
 $$
 \langle \psi|M_1|\psi\rangle=\frac{1}{k}\sum_{q,q'\in S}\sum_{i\in [n/2]}[q=i=q']=\frac{|[n/2]\cap S|}{k} 
 $$
 which is at least $1/k$ if and only if there is an $i \in [n/2]\cap S$. So if we do a $\Qstat$ query with $M_1$ and tolerance $1/(2k)$, the learning algorithm learns if there is an $i \in [n/2]$ such that $i \in S$. Repeat this using a binary search and we will eventually find one element in $S$ using $O(\log n)$ $\Qstat$ queries. Repeat this to find all the elements in $S$, so the overall complexity is $O(k \log n)$. 

    We next give a learning algorithm for $\Cc_3$.
 Consider the $\QSQ$ queries $M_j=\ketbra{e_j}{e_j} \otimes \ketbra{0}{0}$ and $\tau=1/(2n)$.  Then observe that 
 $$
 \langle \psi_x| M_j| \psi_x\rangle = [e_j^\top M x=0]/n,
 $$ which equals $1/n$ if $e_j^\top Mx=0$ and $0$ otherwise, so with tolerance $1/(2n)$, we can learn which is the case. Since $G$ is the generator matrix of a good code, i.e., $G$ has rank $k$, there are $k$ linearly independent rows in $G\in \mathbb{F}_2^{n\times k}$ (say they are $G^{i_1},\ldots,G^{i_k}$). The learning algorithm can perform these $\Qstat$ measurements for all $M_{i_1},\ldots,M_{i_k}$ in order to learn $G^{i_1} x,\ldots,G^{i_k} x$. Since $G^i$s are linearly independent, these $k$ linearily independent constraints on $x$ suffice to learn $x$.
\end{proof}

We next observe that the set of trivial states, i.e., states $\ket{\psi}=C\ket{0^n}$ where $C$ is a constant-depth $n$-qubit circuits, can be learned in polynomial time in the $\QSQ$ model. An open question of this work, and also the works of~\cite{hinsche2022single,sweke23}, is if we can learn the distribution $P_C=\{\langle x|\psi\rangle^2\}_x$ using \emph{classical} $\SQ$ queries. The theorem below shows that if we had direct access to $\ket{\psi}$, one can learn the state and the corresponding distribution $P_C$, using $\QSQ$ queries. In the next section we show that once the depth $d=\omega(\log n)$, these states are hard for $\QSQ$ queries as well. 

\begin{theorem}\label{thm:trivial_states}
    The class of $n$-qubit trivial states  can be learned up to trace distance $\leq \varepsilon$ using $\poly(n,1/{\varepsilon})$ $\Qstat$ queries with tolerance~$\poly(\varepsilon/n)$.
\end{theorem}
\begin{proof}
    Say that the circuit depth is $d$. Using  \cite[Theorem~4]{yu2023learning} it is sufficient to reconstruct all $D:=2^d$-body reduced density matrices up to precision $\frac{\varepsilon^2}{4n}$ with respect to trace distance. Thus, it is sufficient to show that such a tomography can be accomplished with $\Qstat$ queries. To do so, simply query all $4^{D}-1$ non-identity Pauli strings acting on a party $s$ of size $D$ and reconstruct the state as $\hat{\rho}_s = \frac{1}{2^D}(\id + \sum_x \alpha_x P_x)$, where $P_x$ is a non-identity Pauli string and $\alpha_x$ is the response upon querying $P_x$. The schatten 2-norm of the difference between the resulting state and the true reduced density matrix $\rho_s$ must satisfy 
    \begin{align}
        \Vert \rho_s - \hat{\rho}_s \Vert_2^2 & = \Tr[(\rho_s - \hat{\rho}_s)^2]\\
        & = \frac{1}{4^D}\Tr\left[\left(\sum_x (\Tr[P_x\rho_s]-\alpha_x)P_x\right)^2\right]\\
        & = \frac{1}{2^D}\sum_x \left(\Tr[P_x\rho_s]-\alpha_x\right)^2\\
        & < \tau^2 \cdot 2^D,
    \end{align}
    where we have used that Pauli strings (including identity) satisfy $\Tr[P_xP_y]=\delta_{x,y}2^D$. In general, $\TRD(\rho_s,\hat{\rho}_s) \leq 2^{D/2-1}\Vert \rho_s - \hat{\rho}_s\Vert_{2}$. Thus, $\TRD(\rho_s,\hat{\rho}_s) < \tau\cdot 2^{D-1}$. Taking $\tau \leq \frac{\eps^2}{n\cdot 2^{D+1}} = O(\frac{\eps^2}{n})$ yields a tomography with the desired precision. There are $\binom{n}{D} = O(n^{D})$ such reduced density matrices. For each one, we require a constant number of $\Qstat$ queries, each requiring $O(D)=O(1)$ gates. Thus, the overall complexity is $O(n^{D}) \in \poly(n)$ for both query and time complexity.
\end{proof}

\subsection{Hardness of testing purity}
\begin{theorem}
Let $\A$ be an algorithm that upon a input of quantum state $\rho$, with high probability, estimates the purity of $\rho$ with error $< 1/4$ using $\Qstat(\tau)$ queries. Then $\A$ must make at least $2^{\Omega(\tau^22^n)}$ such queries.
\end{theorem}
\begin{proof}
    Suppose we have such an algorithm $\A$. Then we could solve the decision problem of $\Cc = \{U\ketbra{0}{0}U^\dagger \ \vert \ U\in\mathcal{U}(2^n)\}$ (all pure states) versus $\sigma = \frac{1}{2^n}\id$. We prove that this decision problem is hard using a concentration of measure argument similar to the variance method. Drawing pure states from the Haar measure on $\mathcal{U}(2^n)$ yields $\E[U\ketbra{0}{0}U^\dagger] = \frac{1}{2^n}\id$. Upon querying an observable $M$, consider the adversarial response of $\frac{1}{2^n}\Tr[M]$. By Levy's Lemma(~\ref{thm:levy}), most Haar random states cannot deviate much from this average. For our purposes, we are concerned with functions of the form $f(\ket{\psi}) = \Tr[M\ketbra{\psi}{\psi}]$ where $\Vert M \Vert \leq 1$. We immediately observe that such $f$'s have Lipschitz constant $2$ \cite{Popescu_Short_Winter_2006}. By Levy's Lemma~\ref{thm:levy} we have that
    \begin{align}
        \Pr_{U}[\vert \Tr[MU\ketbra{0}{0}U^\dagger] - \frac{1}{2^n}\Tr[M] \vert > \tau] \leq 2\exp \left(-\frac{2^{n+1}\tau^2}{36\pi^3}\right)
    \end{align}
    To conclude, we use Levy's lemma to lower bound $\QSDA_\tau$ (in a manner similar to that of the variance lower bound). Recall that $\QSDA_\tau(\Cc, \sigma)$ is the smallest integer $d$ such that there exists a distribution $\eta$ over $\Qstat$ queries $M$ such that $\forall \rho \in \Cc: \Pr_{M \sim \eta}[\vert \Tr[M(\rho - \sigma)]\vert > \tau] \geq 1/d$. From this definition we have that
\begin{align}
        \frac{1}{d} \leq \Pr_{M\sim \eta}\Pr_{\rho \sim \mu}[\Tr[M(\rho - \sigma)]\vert > \tau] \leq 2\exp\left(-\frac{2^{n+1}\tau^2}{36\pi^3}\right)\ .
    \end{align}
    Thus, $\QSDA_\tau\left(\{U\ketbra{0}{0}U^\dagger\}, \frac{1}{2^n}\id\right) \geq 2^{\Omega(\tau^2 2^n})$ and $\A$ must make at least $2^{\Omega(\tau^2 2^n)}$ queries.
\end{proof}

\subsection{Hardness of the Abelian hidden subgroup problem}
One of the great successes of quantum computing is solving the hidden subgroup problem for Abelian groups, of which Shor's famous factoring algorithm is a consequence. In this problem, we are given query access to a function $f$ on a group $G$ such that there is some subgroup $H \leq G$ satisfying $f$ is constant every left coset of $H$ and is distinct for different left cosets of $H$. How many queries to $f$ suffice to learn $H$? When $G$ is a finite Abelian group, $H$ can be efficiently determined by separable quantum algorithms. One approach which is often used to analyze the general Hidden subgroup problem is the \emph{standard} approach, which we describe now~\cite{wang2010hidden}:
\begin{enumerate}
    \item Prepare the superposition $\frac{1}{\sqrt{\vert G \vert}}\sum_{g\in G} \ket{g}\otimes \ket{0}$ by a Fourier transform over the group~$G$.
    \item Use a single query to prepare the superposition state $\frac{1}{\sqrt{\vert G \vert}}\sum_{g\in G} \ket{g}\otimes\ket{f(g)}$.
    \item Measure the second register and obtain a superposition over elements in some coset with representative $g'$. That is, the algorithm can be viewed as having the state  $\rho_H=\sum_{g'}\ketbra{\psi_{g'H}}{\psi_{g'H}}$ where $\ket{\psi_{g'H}}=\frac{1}{\sqrt{\vert H \vert}}\sum_{g\in g'H}\ket{g}$.
    \item Again apply a quantum Fourier transform and measure the state to obtain an element $g\in H^\perp$, where $H^\perp = \{g\in G \vert \chi_g(H) =1\}$.
\end{enumerate}

Repeating the above procedure $\tilde{O}(\log \vert G \vert)$ times yields a generating set for $H^\perp$ with high probability, allowing one to reconstruct $H$ as well. In fact observe that the above algorithm works even if one just makes \emph{separable} measurements. The state $\rho_H$ in step (3) of the algorithm above is called a \emph{coset state}. Here, we show that solving the Hidden subgroup problem for even Abelian groups is hard when the learning algorithm has access only to $\QSQ$ queries.

Consider the additive group $G = \mathbb{Z}_2^n$. In Simon's problem, a version of the hidden subgroup problem on $\mathbb{Z}_2^n$, the hidden subgroups are of the form $H = \{0,s\}$. While solving Simon's problem is easy using separable quantum measurements, it cannot be readily replicated using $\Qstat$ queries. Intuitively, every $y$ in the orthogonal complement of $s$ is equally likely to be observed upon a computational basis measurement. To see this, note that after discarding the register containing the function value, the resulting mixed states are $\rho_s = \frac{1}{2^{n-1}}\sum_{\overline{x}}\ketbra{\overline{x}}{\overline{x}}$, where $\overline{x}$ is a coset representative and $\ketbra{\overline{x}}{\overline{x}}$ is the projector onto the corresponding coset. Thus, accurately simulating this measurement with $\Qstat$ queries requires exponentially small tolerance $\tau$. The following theorem formalizes this~notion.
\begin{theorem}
    Solving the hidden subgroup problem for the Abelian group $\mathbb{Z}_2^n$ with $\Qstat$ queries of the form $M = M'\otimes \Id$ requires $\Omega(\tau^2 \cdot 2^n)$ many such queries to succeed with high probability.
\end{theorem}
\begin{proof}
    We prove the theorem by a bound on $\QAC_\tau(\Cc, \sigma)$ where $\sigma = \frac{1}{2^n}\id$. Say that $\A$ is an algorithm which solves the hidden subgroup problem with high probability using $\Qstat$ queries of the form $M=M'\otimes \Id$. Then, the queries $\{M_i\}_i$ used by $\A$ imply the existence of queries $\{M_i'\}_i$ where $M_i'\in\C^{2^n}\times \C^{2^n}$ which suffice to identify the coset states $\rho_H = \frac{1}{\vert H \vert}\sum_{\overline{x}}\ketbra{\overline{x}}{\overline{x}}$, where $\overline{x}$ denotes a coset and $\ketbra{\overline{x}}{\overline{x}}$ the projector onto this coset.

    Consider the subset $\Cc_0 \subset \Cc$ of coset states of subgroups of the form $H_s = \{0,s\}$. For such a subgroup $H_s$ the corresponding coset state is $\rho_s = \frac{1}{2^{n-1}}\sum_{\overline{x}}\ketbra{\overline{x}}{\overline{x}}$, where $\{\overline{x}\}$ are a set of $2^{n-1}$ coset representatives and $\ket{\overline{x}} = \frac{1}{\sqrt{2}}(\ket{x} + \ket{x \oplus s})$. If $f$ is a constant function, then $\rho_H = \frac{1}{2^n}\Id$. Thus, the correctness of $\A$ implies the existence of a $\QSQ$ algorithm that can solve the decision problem of $\{\rho_{H=\{0,s\}}\}_H$ versus $\sigma = \frac{1}{2^n}\id$.

   For such a decision problem, $\hat{\rho_s} = 2^n\rho_s - \id$ and $\Tr[\hat{\rho}_{s}\hat{\rho}_{s'}\sigma] = 2^n\Tr[\hat{\rho}_{s}\hat{\rho}_{s'}]-1$. Let $s = s'$. Then $\Tr[\rho_s^2]=2^{-2(n-1)}$ and $\Tr[\hat{\rho}_{s}^2\sigma] = 1$. Now instead consider when $s\neq s'$. For every coset $\overline{x}$ of $H_s$ there exist two cosets $\overline{y}_1$ and $\overline{y}_2$ of $H_{s'}$ with a non-empty intersection (of exactly one element) with $\overline{x}$. Thus, we have that $\Tr[\ketbra{\overline{x}}{\overline{x}}\ketbra{\overline{y}_1}{\overline{y}_1}] = \Tr[\ketbra{\overline{x}}{\overline{x}}\ketbra{\overline{y}_2}{\overline{y}_2}]= \frac{1}{4}$ and
    \begin{align}
        \Tr[\hat{\rho}_s \hat{\rho}_{s'}\sigma] & = 2^n\Tr[\rho_s \rho_{s'}]-1 = \frac{2^n}{2^{2(n-1)}}\sum_{\overline{x},\overline{y}} \vert \braket{\overline{x}}{\overline{y}}-1 = \frac{2^n}{2^{2(n-1)}} \sum_{\overline{x}}\frac{1}{2}-1 = 0\ .
    \end{align}
    For any subset $\Cc'\subseteq \Cc_0$ we thus have that $\gamma(\Cc',\sigma) = \frac{1}{\vert \Cc' \vert}$. If $\vert \Cc'\vert < \frac{1}{\tau}$ then $\gamma(\Cc',\sigma) > \tau$. Note that $\vert \Cc_0 \vert = 2^n-1$ and thus $\kappa_\tau^\gamma-\textsf{frac}(\Cc_0,\sigma) = \Theta(\frac{1}{\tau 2^{n}})$ and $\QAC_\tau(\Cc,\sigma) = \Omega(\tau\cdot 2^n)$. By Theorem~\ref{thm:mainlowerboundonqsq} we  have that $\QSD_\tau^\delta(\Cc,\sigma) \geq (1-2\delta)\QAC_{\tau^2}(\Cc,\sigma) = \Omega(\tau^2\cdot 2^n)$.
\end{proof}

Thus, any $\QSQ$ algorithm for solving the hidden subgroup problem on $\mathbb{Z}_2^n$ must depend non-trivially on the register holding the function value. This is in contrast to the standard Fourier sampling method which has no dependence on the function register.

\begin{remark}
    The average correlation argument above also implies that learning coset state below trace distance $\frac{1}{2}$ with high probability requires $\Omega(\tau^2 \cdot 2^n)$ $\Qstat$ queries of tolerance $\tau$.
\end{remark}
\begin{proof}
    Note that the trace distance between $\rho_H$ for $H=\{0,s\}$ and $\frac{1}{2^n}\id$ is $\frac{1}{2}$. Using Lemma~\ref{lem:learn_decide_2}, we have that $QSQ^{1/2, \delta}_{\tau}(\Cc) \geq \QSD_\tau^\delta(\D, \frac{1}{2^n}\id)$, where $\D = \{\rho_H \vert H = \{0,s\}\}$. From section 3 we know that $\QSD_\tau^\delta(\D, \frac{1}{2^n}\id) \geq (1-2\delta)\QAC_{\tau^2}(\Cc, \frac{1}{2^n}\id)$. The average correlation argument above yields that $\QAC_{\tau^2}(\Cc, \frac{1}{2^n}\id) = \Omega(\tau^2 \cdot 2^n)$, thus proving the claim.
\end{proof}
\subsection{Hardness of shadow tomography}
In \cite{DBLP:conf/focs/ChenCH021} the authors derive lower bounds on the sample complexity of shadow tomography using separable measurements. Recall that in shadow tomography, given copies of $\rho$, the goal of a learner is to predict the expectation value $\Tr[O_i \rho]$ of a collection of known observables $\{O_i\}_i$ up to error $\varepsilon$. To prove these lower bounds the authors construct a many-vs-one decision task where~$\sigma = \id/{2^n}$ and 
\begin{align}
    \Cc = \left\{\rho_i = \frac{\id+3\varepsilon O_i}{2^n}\right\}\ .
\end{align}
Assuming that $\Tr[O_i] =0 $ and $\Tr[O_i^2]=2^n$ for all $O_i$, then an algorithm which solves the shadow tomography problem with high probability also solves the decision problem. Thus, a lower bound on the latter is also a lower bound on the sample complexity of shadow tomography.

\begin{theorem}\label{thm:CCHL}
    Any algorithm that uses $\Qstat(\tau)$ queries and predicts $\Tr[P\rho]$ up to error $\varepsilon$ for all non-identity Pauli strings $P$ with high probability requires $\Omega(\tau^2 \cdot 2^{2n}/\varepsilon^2)$ queries. 
\end{theorem}
\begin{proof}
    We prove the theorem using a bound on $\QAC_\tau(\Cc, \sigma)$ where $\sigma = \frac{1}{2^n}\id$. For convience we label the states from the many-vs-one decision task as $\rho_i$ where $i\in [4^{n}-1]$. For such a $\sigma$ we further have that $\Tr[\hat{\rho}_i\hat{\rho}_j\sigma] = 2^{n}\Tr[\rho \rho']-1$. By the orthogonality of Pauli strings, $\Tr[\hat{\rho}_i\hat{\rho}_j\sigma] = 9\varepsilon^2\delta_{i,j}$. For any subset $\Cc'\subseteq \Cc$ we thus have that $\gamma(\Cc', \sigma) = \frac{9\varepsilon^2}{\vert \Cc' \vert}$. If $\vert \Cc' \vert < \frac{9\varepsilon^2}{\tau}$ then $\gamma(\Cc', \sigma) > \tau$. Thus, $\kappa_\tau^\gamma-\textsf{frac}(\Cc,\sigma) = \Theta(\varepsilon^2\cdot \tau^{-1}\cdot 2^{2n})$ and $\QAC_\tau(\Cc, \sigma) = \Omega(\frac{\tau\cdot 2^{2n}}{\eps^2})$. Using Theorem~\ref{thm:mainlowerboundonqsq} we know that $\QSD_\tau^\delta (\Cc,\sigma) \geq (1-2\delta)\QAC_{\tau^2}(\Cc,\sigma) = \Omega(\tau^2\cdot 2^{2n}/\eps^2)$.
 \end{proof}

In \cite{DBLP:conf/focs/ChenCH021} the authors derive a lower bound of $\Omega(2^n/\eps^2)$ for the same task. They further show an upper bound of $O(n2^n/\eps^2)$ as well. Our result essentially says that shadow tomography benefits from more than just estimating expectation values. With only $\Qstat$ queries, the nearly optimal algorithm is to simply query every Pauli string.

\subsection{Learning quantum biclique states}

An influential work of Feldman et al.~\cite{feldman2017statistical} considers  the planted biclique problem. The goal here is to learn the class of distributions each indexed by subsets $S \subseteq \lbrace 1,2 \ldots ,n \rbrace$. For every $S$, the distribution $D_S$ is defined as follows 
\[ 
D_S(x)=\begin{cases} 
      \frac{k/n}{2^{n-k}}+\frac{1-k/n}{2^n} & x\in 1_S\times \01^{n-k} \\
      \frac{1-k/n}{2^n} &  x\notin 1_S\times \01^{n-k},
   \end{cases}
\]
where above $1_S\times \01^{n-k}$ is the set $\{x\in \01^n: x_S=1_S\}$. 

A natural way of generalizing problems over distributions to quantum statistical queries is to consider coherent encodings of distributions, i.e., for a given distribution $D$ over $X$, we define a quantum state $\ket{\psi} = \sum_x \sqrt{D(x)}\ket{x}$. Classical $\Stat$ queries then correspond to $\Qstat$ with \emph{diagonal} observables and a natural question is, how much can coherent examples help? In what follows, we first show that for the task of distinguishing two coherent encodings, there can be at most a quadratic gap between the precision that is tolerated by $\Qstat$ and $\Stat$ queries. We use this to show that, for some choice of parameters, there are large gaps between the classical and quantum statistical query complexity of the $k$-biclique problem.
 We demonstrate below that $\Qstat$ measurements can help significantly in certain regimes of tolerance.

\begin{lemma}
For large enough $n$ and $k\geq 2\log n$, the $k$-planted biclique problem with coherent encodings can be solved with statistical quantum algorithm that makes at most $\binom{n}{k}$ $\Qstat \left(\sqrt{k/n}\right)$ queries, but cannot be solved by any algorithm that makes $\Stat\left(\sqrt{k/n}\right)$ queries. 

\end{lemma}

\begin{proof}
First observe that $\TRD (\ket{\psi_S}, \ket{+^n}) = \sqrt{1-|\braket{\psi_S}{+^n}|^2}$, and
\begin{align}
    \braket{+^n}{\psi_S} &= \left(\sqrt{\frac{k}{n} + \frac{1-k/n}{2^k}} - \sqrt{\frac{1-k/n}{2^k}} \right) \frac{1}{\sqrt{2^k}} + \sqrt{1 - \frac{k}{n}}.
\end{align}
Define $\ket{\phi} = (\ket{+} + \ket{\psi_S})/\sqrt{2+2\langle \psi_S|+\rangle}$ and $\ket{\phi^\intercal}  (\ket{+} - \ket{\psi_S})/\sqrt{2-2\langle \psi_S|+\rangle}$.
The optimal distinguishing $\Qstat$ query between $\ket{+^n}$ and $\ket{\psi_S}$ is the difference between projectors on the state $\ket{P_S} = \frac{{\ket{\phi} + \ket{\phi^\intercal}}} {\sqrt{2}}$ and its orthogonal complement in the span of $\ket{+^n}$ and $\ket{\psi_S}$ (see for example~\cite[Theorem~3.4]{watrous2018theory}). Call this measurement $M_S$ and notice that it is implementable by a $k$-qubit controlled rotation. A (possibly inefficient) quantum algorithm for detecting the planted clique would query  $\Qstat(\tau)$ oracle with $M_S$ for every subset $S \subseteq [n]$ of cardinality $|S| = k$. 
From Lemma \ref{Lemma:norms} and optimality of the measurement, we know that $|\Tr(M_S (\psi_S - \psi_0))| = 2\TRD(\psi_S, \psi_0)$. It follows that 
as long as
$\tau \leq \TRD(\ket{\Psi_S}, \ket{+^n})$, such algorithm succeeds.  We now bound $\TRD(\ket{\Psi_S}, \ket{+^n})$. To that end, observe that:
\begin{align}
  \left(\sqrt{\frac{k}{n} + \frac{1-k/n}{2^k}} - \sqrt{\frac{1-k/n}{2^k}} \right) \frac{1}{\sqrt{2^k}} \leq \frac{1}{\sqrt{2^{k+1}}} \sqrt{1- \frac{k}{n}}.
\end{align}
from which we have that:
\begin{align}
   \left( 1 + 2^{-(k+1)/2}\right) \sqrt{1-\frac{k}{n}}\geq \braket{+^n}{\psi_S} \geq \sqrt{1-\frac{k}{n}},
\end{align}
and\footnote{Using 
$ \sqrt{1+3 \times 2^{-(k+1)/2}} \geq 1 + 2^{-(k+1)/2}$ for all $k \geq 1 $.}
$\sqrt{\frac{k}{n} } \geq \TRD(\psi_0, \psi_S)   \geq \sqrt{ \frac{k}{n} - \frac{4}{2^{k/2}} }$. For $k \geq 2 \log n$, $n \geq 5$ and $\tau \leq\sqrt{\frac{2\log(n/4)}{n}}$,  the planted biclique can be detected by at most $n \choose k$ $\Qstat(\tau)$ queries.
 On the other hand, the $k$-planted biclique problem has $\TVD(D, D_i) = \frac{k}{n} \left( 1 - 2^{-k} \right)$ for all $D_i \in \mathcal{D}_D$, from which 
$\TVD(D, D_0) = \frac{k}{n}(1-2^{-k}) < \frac{k}{n}$.
It follows that: 
\begin{equation*}
 \max_{\phi, |\phi| \leq 1} \Pr_{D \sim \mathcal{D}} \left[ |D[\phi] - D_0[\phi]|  \geq 2\tau\right]   \leq  \Pr_{D \sim \mathcal{D}} \max_{\phi, |\phi| \leq 1} \left[   |D[\phi] - D_0[\phi]| \geq 2\tau \right] = \Pr_{D \sim \mathcal{D}} \left[ \TVD(D,D_0) \geq \tau \right].
\end{equation*}
For $\tau =k/n$, we have $ \Pr_{D \sim \mathcal{D}} \left[ \TVD(D,D_0) \geq \frac{k}{n} \right]  = 0$,
which means that the clique state is undetectable by any $\Stat(\tau)$ query (an adversarial oracle can output an outcome consistent with uniform distribution and succeed at all times).  For $k \geq 2 \log n$  and large enough $n$, the statistical queries have better tolerance than the quantum queries. It follows that for $k \geq 2 \log n$, $n \geq 72$ and $\tau = \sqrt{\frac{2\log(n/4)}{n}}$, the $k$-planted biclique problem cannot be solved by a $\Stat(\tau)$ algorithm, but can be solved with an algorithm that can makes $\Qstat(\tau)$ queries.
\end{proof}

\subsection{Hardness of Learning Approximate Designs}
In this class we show that the class of quantum states that forms an approximate $2$-designs are hard to learn in the $\QSQ$ model. 
\begin{theorem}\label{thm:approx_design}
    Let $\Cc$ be an ensemble of states forming a $\eta$-approximate $2-$design where $\eta = O(2^{-n})$. Learning states from $\Cc$ with error $\leq 1/3$ in trace distance requires $\Omega(\tau^2 \cdot 2^n)$ $\Qstat(\tau)$~queries.
\end{theorem}
\begin{proof}
    We prove the theorem by showing that the variance of $\{\Tr[M\rho]\}_{\rho \in \Cc}$ for any such design must be exponentially small. By the definition of an approximate design, we have that
    \begin{align}
        d_{\Tr}\left(\E_{\rho \sim \Cc}[\rho], \frac{1}{2^n}\id\right) \leq \eta, \quad d_{\Tr}\left(\E_{\rho \sim \Cc}[\rho^{\otimes 2}], \frac{1}{4^n + 2^n}(\id + \SWAP)\right) \leq \eta \ ,
    \end{align}
    where $\frac{1}{2^n}\id$ and $\frac{1}{4^n+2^n}(\id + \SWAP)$ are respectively the first and second moments of the unitary Haar measure. For any observable $M$, by the definition of trace distance we have that
    \begin{align}
        \left\vert \Tr[M(\frac{1}{2^n}\id - \E_{\rho \sim \Cc}[\rho])] \right\vert \leq 2\eta, \quad 
        \left\vert \Tr[M(\frac{1}{4^n+2^n}(\id+\SWAP) - \E_{\rho \sim \Cc}[\rho^{\otimes 2}])] \right\vert \leq 2\eta\ .
    \end{align}
    Thus, $\Var_{\rho \sim \Cc}(\Tr[M \rho]) \leq \Var_{\rho \sim \mathcal{U}(2^n)}(\Tr[M \rho]) + O(2^{-n})$. We now show that $\Var_{\rho \sim \mathcal{U}(2^n)}(\Tr[M \rho]) = O(2^{-n})$ for any $\Vert M \Vert \leq 1$.
    \begin{align}
        \mathop{\Var}_{\rho \sim \mathcal{U}(2^n)}(\Tr[M \rho]) & = \frac{1}{4^n+2^n}\Tr[M^{\otimes 2}(\id + \SWAP)] - \frac{1}{4^n}\Tr[M]^2\\
        & = \frac{1}{4^n+2^n}\Tr[M^2] - \frac{1}{2^n(4^n+2^n)}\Tr[M]^2\ ,
    \end{align}
    where we have used that fact that $\Tr[M^{\otimes 2}] = \Tr[M]^2$ and $\Tr[M^{\otimes 2} \SWAP] = \Tr[M^2]$. As $\Vert M \Vert \leq 1$ we have that $\Tr[M^2] \leq 2^n$. Thus, taking $M$ to have $2^{n-1}$ eigenvalues equal to $+1$ and $2^{n-1}$ equal to~$-1$ maximizes the variance yielding
    \begin{align}
        \mathop{\Var}_{\rho \sim \Cc}(\Tr[M \rho]) & \leq \mathop{\Var}_{\rho \sim \mathcal{U}(2^n)}(\Tr[M \rho]) + O(2^{-n})  \leq \frac{2^n}{4^n+2^n} +O(2^{-n}) = O(2^{-n})\ .
    \end{align}
    To invoke Theorem~\ref{thm:mainlowerboundonqsq} and Lemma~\ref{Lemma:LearningToDeciding} we first note that all $\rho' \in \Cc$ are far from $\E_{\rho \sim \Cc}[\rho]$ in trace distance. This follows from triangle inequality: 
    $$
    d_{\Tr}\left(\rho' , \E_{\rho \sim \Cc}[\rho]\right) \geq d_{\Tr}\left(\rho', \frac{1}{2^n}\id\right)-d_{\Tr}\left(\frac{1}{2^n}\id, \E_{\rho \sim \Cc}[\rho]\right)\geq \frac{2^n-1}{2^n}-O(2^{-n}).
    $$
    Fixing $\varepsilon = 1/3$ and $\tau = 1/\poly(n)$ there is an $n_0$ such that for all $n\geq n_0$ we have that $d_{\Tr}(\rho' , \E_{\rho \sim \Cc}[\rho]) > 2(\tau + \varepsilon)$. Using Lemma~\ref{Lemma:LearningToDeciding} we thus have that learning states from $\Cc$ requires $\Omega(\tau^2 \cdot 2^{n})$ $\Qstat$ queries.
\end{proof}


\section{Further applications}
\label{sec:app}
\subsection{Error mitigation}
In this section, we show how to use our $\QSQ$ lower bound to 
 resolve an open question posed by Quek et al.~\cite{quek2022exponentially}. Therein the authors consider two forms of quantum error mitigation, which they call \emph{strong} and \emph{weak} error mitigation. We first describe these two models before stating our~result. 

\begin{definition}[Weak Error Mitigation]
An $(\varepsilon, \delta)$ weak error mitigation algorithm $\mathcal{A}$ takes an input a series of observables $\{O_1,\ldots, O_m\}$ satisfying $\Vert O_i \Vert \leq 1$ and outputs a set of values $\{\alpha_1,\ldots, \alpha_m\}$ such that with probability at least $1-\delta$ we have that
\begin{align}
    \left\vert \Tr[O_i \rho] - \alpha_i \right\vert \leq \varepsilon \ .
\end{align}
    
\end{definition}

\begin{definition}[Strong Error Mitigation]
    An $(\varepsilon, \delta)$-strong error mitigation algorithm $\mathcal{A}$ outputs a bitstring $z$ sampled from a distribution $P$ such that, with probability at least $1-\delta$, $\TVD(P,P_\rho)\leq \varepsilon$. Here $P_\rho$ is the distribution on the computational basis induced by the state $\rho$, i.e. $P_\rho(x) = \Tr[\ketbra{x}{x}\rho]$.
\end{definition}

In both cases, we assume that the algorithm is given classical descriptions of both the preparation and noise channels resulting in $\rho$. Further, $\A$ can make measurements on multiple copies of $\rho$ at once.

\begin{remark}
    For some forms of error mitigation it may be interesting to consider not just allowing the algorithm to query the circuit $U_{\mathcal{C}}$ but also modified circuits $U_{\mathcal{C'}}$. However this can be subsumed into the framework of weak error mitigation as given. To return an estimate of $O$ for $U_{\mathcal{C}'}$ the algorithm returns an estimate of $U_{\mathcal{C}}U_{\mathcal{C}'}^\dagger O U_{\mathcal{C}'}U_{\mathcal{C}}^\dagger$ from the original circuit.
\end{remark}

In \cite{quek2022exponentially} the authors show that strong error mitigation implies weak error mitigation for local observables. They then prove a partial converse and show that for a restricted family of observables weak error mitigation cannot recover strong error mitigation (for polynomial-sized inputs). The question of an unconditional separation is left open. Here we will show that Theorem~\ref{thm:mainsep} closes this open question and implies that weak error mitigation with polynomial numbers of observables does not suffice to recover strong error mitigation. First, note that by definition weak error mitigation outputs $\QSQ$ queries with tolerance $\tau = \varepsilon$. To match our notation, we will continue by using $\tau$ instead of $\varepsilon$. This is the equivalent of \cite[Lemma~5]{quek2022exponentially}. Next we have the following theorem from their work.
\begin{theorem}{\cite[Theorem 5]{quek2022exponentially}}\label{thm:distalg}
    For a class of distributions $\mathcal{Q} = \{q_1,\ldots, q_k\}$ and $\varepsilon, \delta > 0$ there is an algorithm which takes $O(\frac{\log \vert \mathcal{Q} \vert}{\varepsilon^2})$ samples from a target distribution $p$ (not necessarily in $\mathcal{Q}$) and outputs a $q^* \in \mathcal{Q}$ such that
    \begin{align}
        \TVD(p,q^*) \leq 3\min_{i\in [k]} \TVD(p,q_i) + \varepsilon \ .
    \end{align}
\end{theorem}

With these tools we can now prove a separation between strong and weak error mitigation.

\begin{theorem}\label{thm:weakvstrongEM}
    Let $\A$ be an algorithm that takes as inputs the estimates for weak error mitigation with $\tau = 1/\poly(n)$ and outputs $O(n^2)$ samples from some distribution $P$ such that $  \TVD(P,P_\rho)< 1/20$ with high probability. Then, $\A$ requires estimates of $\Omega(\tau^2\cdot 2^{n/2})$ distinct observables.
\end{theorem}
\begin{proof}
   We show that such samples would give one the ability to exact learn quadratic polynomial states with polynomial $\QSQ$ queries, contradicting Theorem~\ref{thm:mainsep}. Let $P_A$ denote the distribution on the computational basis induced by $\ket{\psi_A} = \frac{1}{\sqrt{2^n}}\sum_x \ket{x}\otimes \ket{x^\top A x}$. For this concept class we can directly identify an example state with the distribution it induces on the computational basis and vice versa.
    Let's assume that such an algorithm $\A$ does exist. Using Theorem~\ref{thm:distalg}, and noting that $\log \vert \Cc \vert = \Theta(n^2)$,
   we can obtain a $P_{B}$ such~that 
    \begin{align}
        \TVD(P,P_{B}) \leq 3\min_{B'\in\mathcal{C}} \TVD(P,P_{B'}) + 1/20.
    \end{align}
    By the assumption upon $P$, we have that $\TVD(P, P_{A}) < 1/20$, where $A$ is the true concept. Thus, $\TVD(P,P_B) < 1/5$. For $A\neq B$ we have that
    \begin{align}
        \TVD(P_A, P_B) & = \frac{1}{2}\sum_{x,y}\vert P_A((x,y)) - P_B((x,y))\vert\\
        & = \frac{1}{2^{n+1}}\sum_{x,y} \vert \delta[x^\top A x = y] - \delta[x^\top B x =y]\vert = \Pr_{x\sim\01^n}[x^\top A x \neq x^\top B x] \geq 1/4,
    \end{align}
    where the last inequality follows from Fact~\ref{fact:zippel}.  For $A \neq B$, by the triangle inequality, we have that
    \begin{align}
        1/4 \leq \TVD(P_A, P_B) & \leq \TVD(P,P_A) + \TVD(P,P_{B}) < 1/4\ .
    \end{align}
    Thus we must have that $A^* = B$ and the true distribution (and function/state) can be recovered. This implies that the inputs to $\A$ could have been used as $\Qstat$ queries to solve the approximate state learning problem of Theorem~\ref{thm:mainlowerboundonqsq}. By the hardness of this problem $\A$ requires $\Omega(\tau^2\cdot 2^{n/2})$ distinct observables as inputs.
\end{proof}

\subsection{Learning distributions}
In this section, we consider the following setup of statistical query learning that was considered in the work of~\cite{hinsche2022single}. Let $U$ be a unitary and consider the induced distribution $P_U$ on the computational basis, i.e.,
$$
P_U(x)=\langle x| U|0^n\rangle^2.
$$
In~\cite{hinsche2022single,sweke23} they considered learning algorithms that were given access to the following: for $\phi:\01^n\rightarrow [-1,1]$ and $\tau\in [0,1]$,
$$
\Stat: (\phi,\tau)\rightarrow \alpha_{\phi} \in \Big[\mathop{\E}_{x\sim P_U}[\phi(x)]+\tau,\mathop{\E}_{x\sim P_U}[\phi(x)]-\tau\Big].
$$
The goal of the learning algorithm is to learn $P_U$ upto total variational distance $\leq \varepsilon$ by making $\poly(n)$ many $\Stat$ queries each with tolerance $\tau=1/\poly(n)$.  Hinsche et al.~\cite{hinsche2022single}  showed the hardness of learning the distribution $P_U$ when $U$ is a Clifford circuit of depth $\omega(\log n)$ and recently Nietner et al.~\cite{sweke23} showed that if $U$ is a depth-$\Omega(n)$ circuit where each gate is picked from $U(4)$, then $P_U$ is not learnable using just $\Stat$ queries. 

In this section we consider a stronger question. One can also just directly look at the quantum state $\ket{\psi_U}=U\ket{0^n}$ and ask how many $\Qstat$ queries of the form
$$
\Qstat: (M,\tau)\rightarrow \alpha_{M} \in \Big[\langle \psi_U|M|\psi_U\rangle+\tau,\langle \psi_U|M|\psi_U\rangle-\tau\Big].
$$
suffice to learn $P_U$ upto small trace distance? Note that the learning model in~\cite{hinsche2022single,sweke23} is \emph{a strict restriction} of this model, cause one could just consider $M=\sum_x \phi(x)\ketbra{x}{x}$, then 
$$
\langle \psi_U|M|\psi_U\rangle=\sum_x \phi(x)\langle x|U|0^n\rangle^2=\sum_x \phi(x)P_U(x)=\mathop{\Exp}_{x\sim P_U}[\phi(x)],
$$
which is precisely $\alpha_\phi$. To this end, we first generalize~\cite{hinsche2022single} in the following theorem. 

\begin{theorem}\label{thm:smallckt_distrib}
    For constant $\alpha\in (0,1)$, there is a family of $n$-qubit circuits consisting of $\{\textsf{Had},\textsf{X},\textsf{CX}\}$ gates of depth $d=(\log n)^{1/\alpha}$ and size $d^2$ that requires $2^{\Omega(d)}$ $\Qstat$ queries to learn the output distribution in the computational basis to error $\leq 0.00125$ in total variational~distance.
\end{theorem}
\begin{proof}
    Consider the padded states $\ket{\psi_A}\otimes\ket{0}^{\otimes k(n)}$ we considered in Theorem~\ref{thm:smallckt} where $\{\ket{\psi_A}=\frac{1}{\sqrt{2^n}}\sum_x \ket{x,x^\top A x}\}_{A}$. 
    Using Fact~\ref{lem:statetodistribution} learning the output distributions of these states below total variational distance $0.00125$ implies the existence of an algorithm learning the states up to trace distance~$0.05$. However, we know that doing so requires $2^{\Omega(d)}$ queries using Theorem~\ref{thm:smallckt}. Thus, learning the output distributions requires at least $2^{\Omega(d)}$ $\Qstat$ queries as well.
\end{proof}

We next prove a generalization of~\cite{sweke23}. Before that, we need the following result.
\begin{theorem}{\cite[Theorem~36]{sweke23}}
\label{thm:sweke}
        There exists a $d = O(n)$ such that for any circuit depth $d' \geq d$ and any distribution $Q$ over $\{0,1\}^n$, we have that
        \begin{align*}
            \Pr_{U \sim \mu_{d'}}\left[\TVD(P_U, Q) \geq \frac{1}{225}\right] \geq 1-O(2^{-n}),
     \end{align*}
            \end{theorem}
    where $\mu_{d'}$ indicates the uniform distribution over circuits of depth $d'$.

    \begin{theorem}
     Let $\A$ be an algorithm that makes $\Qstat(\tau)$ queries and with high probability learns the output distributions of $O(n)$-depth random circuits, to error $\leq 1/225$ in total variational distance, then $\A$ must make $\Omega(\tau^2 \cdot 2^n)$ such many queries.
 \end{theorem}
\begin{proof}
     This is a generalization of~\cite[Theorem~6]{sweke23} and follows by a similar analysis to their lower bound. For $d \geq 3.2(2+\ln 2)n+\ln n$ we have that the uniform distribution over depth $d$ random circuits is a $2^{-n}$ approximate 2-design \cite{sweke23}. Like we saw in the proof of Theorem~\ref{thm:approx_design}, for every observable $\Vert M \Vert \leq 1$, we have that $\Var_{\rho \sim \Cc}(\Tr[M\rho]) = O(2^{-n})$. We proceed with the same adversarial lower bound. Upon making a $\Qstat$ query with observable $M$, the adversary responses with $\E_{\rho \sim \Cc}[\Tr[M \rho]]$. Using Chebyshev's inequality,
    \begin{align}
    \label{eq:concentrationvariance}
        \Pr_{\rho \sim \Cc}\left[\left\vert \Tr[M\rho] - \Tr[M\mathop{\Exp}_{\rho \sim \Cc}[\rho]] \right\vert > \tau \right] = \tau^{-2} \cdot 2^{-n} .
    \end{align}
    While the proof could continue using Theorem~\ref{thm:mainlowerboundonqsq} (by reducing to the many-vs-one decision problem), it is more direct to note that the above inequality implies that every deterministic algorithm cannot identify the correct $\rho \in \Cc$  for a  large fraction of states in $\Cc$. 
    
    In particular,  say $\A$ is a deterministic algorithm that outputs an estimate of a distribution $Q$ such that $\TVD(P_\psi, Q) < 1/225$ and uses at most $t$ $\Qstat(\tau)$  queries. Using Eq.~\eqref{eq:concentrationvariance}, there is a fraction of $\Cc'$ of measure at least $1-O(t\cdot \tau^{-2}\cdot  2^{-n})$ that are consistent with $\Tr[M_i \E_{\rho \sim \Cc}[\rho]]$ for the $\Qstat(\tau)$ queries $\{M_1,\ldots,M_t\}$ made by $\A$. Since $\A$ is deterministic, it must output the same distribution $Q$ for all $\rho \in \Cc'$. We now use Theorem~\ref{thm:sweke} to claim that there is a large set of states that are both consistent with $\Tr[M_i \E_{\rho \sim \Cc}[\rho]]$ and far from $Q$ in total variational distance.
    \begin{align*}
        \Pr_{\rho \sim \Cc}\left[\rho \in \Cc' \land \TVD(P_\rho, Q) \geq 1/225\right] & = 1-\Pr_{\rho \sim \Cc}\left[\rho \notin \Cc' \lor \TVD(P_\rho, Q) < 1/225\right] \geq 1 - O(t\cdot \tau^{-2}\cdot 2^{-n})\ ,
    \end{align*}
    where the inequality follows from the union bound, Theorem~\ref{thm:sweke}, and the concentration of measure shown above. Thus, there is a set $\Cc''$ of measure at least $1-O(t\cdot \tau^{-2}\cdot  2^{-n})$ that is both consistent with $\Tr[M_i \E_{\rho \sim \Cc}[\rho]]$ for all queries $M_i$ and also $\TVD(P_\rho, Q) \geq 1/225$ for all $\rho \in \Cc''$. Upon the input of a $\rho\in\Cc''$, $\A$ fails to provide a distribution $Q$ such that $\TVD(P_\rho, Q) < 1/225$. For any constant success probability $1-\delta$ there is then some $n_\delta$ such that for $n \geq n_\delta$ $\A$ must fail on a set of measure strictly greater than $\delta$. Using Yao's Principle, thus any randomized algorithm using $t\in\poly(n)$ $\Qstat$ queries must fail with probability strictly greater than $\delta/2$ for $n$ sufficiently large. Thus, $\A$ must use $t= \Omega(\tau^{2}\cdot 2^{n})$ $\Qstat$ queries.
\end{proof}

\bibliographystyle{alpha}
\bibliography{refs}

\appendix

\section{Upper and Lower Bounds on Search Problems}\label{app:qssd}
Here we use $\QSDA$ to extend to statistical dimension to search problems. First, a definition of a general search problem over quantum states:

\begin{definition}[Quantum Search Problem]
    Let $\Cc$ be a closed set of quantum states and $\mathcal{F}$ some set. Then a decision problem $\Z$ is a mapping $\Z : \Cc \rightarrow 2^{\mathcal{F}}$\footnote{We implicitly assume that $\Z$ maps concepts to measurable subsets of $\mathcal{F}$ with respect to some $\sigma$-algebra on $\mathcal{F}$.}. An algorithm $\A$ is said to solve the search problem if upon the input of a state $\rho \in \Cc$ it returns $f \in \Z(\rho)$.
\end{definition}

Learning quantum states can be cast in this framework using the mapping $\Z(\rho) = \{\sigma\ \vert \ \TRD(\rho,\sigma) \leq \eps\}$. In fact, all of the lower bounds for learning in the main text could be recast in this framework. However, this was over kill for our goals in the main text. Lastly, much more than learning can be cast in this framework. For example, solving the hidden subgroup problem is a search problem.

We now define a quantity we dub the \textit{quantum search statistical dimension}, or $\QSSD$. This is a natural quantization of that given in \cite{Feldman16}, using $\QSDA$ which we just developed.

\begin{definition}[Quantum Search Statistical Dimension]
    Let $\Z$ be a search problem over a closed set of states $\Cc$ and a solution set $\mathcal{F}$. And by $\mathcal{S}^\mathcal{F}$ we denote the space of probability distributions over $\mathcal{F}$. Then, for success probability $\alpha$ the quantum search statistical dimension is:
    \begin{align}
        \QSSD_\tau^\alpha(\Z) = \sup_{\sigma} \inf_{\mu \in \mathcal{S}^{\mathcal{F}}} \QSDA(\Cc \backslash \Z_\alpha(\mu),\sigma)\ ,
    \end{align}
    where $\Z_\alpha(\mu) := \{\rho\in\Cc \ \vert \ \mu(\Z(\rho) \geq \alpha \}$.
\end{definition}

Think of $\mu$ as representing a distribution of solutions an algorithm outputs. Thus, $\QSSD$ represents the hardness of distinguishing the concept class from a reference state given that the algorithm responds with a solution $f$ drawn from $\mu$ upon receiving queries consistent with $\sigma$. This is formalized in the following theorem, whose proof quantizes \cite[Theorem~4.9]{Feldman16}.

\begin{theorem}\label{thm:qssd_lb}[Searching is as hard as deciding]
    For any quantum search problem $\Z$, $\tau >0$, and success probabilities $1-\delta > \alpha > 0$, we have that
    \begin{align}
        \QSQ_\tau^\delta(\Z) \geq \left(1 - \frac{\delta}{1-\alpha}\right) \cdot \QSSD_\tau^\alpha(\Z)\ .
    \end{align}
\end{theorem}
\begin{proof}
 We show that the existence of an algorithm solving the search problem also implies the existence of a cover. Say that $\A$ solves $\Z$ with probability at least $1-\delta$ using at most $q$ queries and let $d = \QSSD_\tau^\alpha(\Z)$. From the definition of $\QSSD_\tau^\alpha$ there is a $\sigma$ such that $\QSDA(\Cc \backslash \Z_\alpha(\mu),\sigma)$ is arbitrarily close to $d$ for any distribution $\mu$ on the solution space. Further, assume that $\A$ receives $\Tr[M\sigma]$ upon querying $M$ whenever this is a valid response. Let $f$ be a random variable (with distribution $\mu$) corresponding to the output of $\A$ in such a scenario upon receiving responses $\{\Tr[M_i \sigma]\}_i$ and let
    \begin{align}
        p_\rho = \Pr_{\A}\left[\exists i\in [q], \ \vert \Tr[M_i(\rho-\sigma)] \vert > \tau\right],
    \end{align}
    which is the probability over the randomness of $\A$ that it can distinguish $\rho$ from $\sigma$. Say that $\rho \notin \Z_\alpha(\mu)$. If $\Tr[M_i\sigma]$ was a valid answer for all queries, then $\A$ draws a solution according to $\mu$. However, by construction this solution is in $\Z(\rho)$ with probability strictly smaller than $\alpha$. However, we assume that $\A$ fails with probability at most $\delta$. Thus, we have that $(1-p_\rho)(1-\alpha) \leq \delta$, implying that $p_\rho \geq 1 - \frac{\delta}{1-\alpha}$.     Let $\hat{M}$ be a random variable with a pdf constructed by taking a uniform average of the distributions of the queries $M_i$ made by the algorithm (which are random variables) assuming that all queries are given responses $\Tr[M\sigma]$\footnote{Thus, the distribution over queries is not dependent on the input state $\rho$.}. Then we have that
    \begin{align}
        \forall \rho \in \Cc, \ \Pr[\vert \Tr[\hat{M}(\rho-\sigma) \vert > \tau] \geq \frac{p_\rho}{q} \geq \frac{1-\frac{\delta}{1-\alpha}}{q}\ .
    \end{align}
    By Lemma~\ref{Lemma:RandomCover} this implies that $\QSDA_\tau(\Cc\backslash\Z_\alpha(\mu),\sigma) \leq \frac{q}{1-\frac{\delta}{1-\alpha}}$ proving the theorem statement.
\end{proof}

Remarkably, $\QSSD$ also \textit{upper bounds} query complexity, which we prove below. Our proof quantizes the proof by Feldman~\cite[Theorem~4.10]{Feldman16}. 
\begin{theorem}\label{thm:qssd_ub}[Searching is not much harder than deciding]
   For any quantum search problem $\Z$, $\tau > 0$, $0< \delta < 1$, and $0 < \beta, \alpha < 1 $ such that $1-\delta= \alpha-\beta$, we have that:
    \begin{align}
        \QSQ_{\tau / 3}^{\delta}(\Z) = O\left(\QSSD_\tau^\alpha(\Z) \cdot \frac{n}{\tau^2}\cdot \log\left(\frac{n}{\tau^2\cdot \beta}\right)\right)\ .
    \end{align}
\end{theorem}
\begin{proof}
    Let $\rho\in\Cc$ be the true state. The online learning algorithm of \cite{aaronson2018online} states that one can update a reference state $\sigma_t$ such that the total regret is upper bounded by $2L\sqrt{(2 \ln 2) Tn}$, where $L$ is the Lipschitz constant of the loss function and $T$ is the total number of updates. Let $E_t$ be the POVM element the online learning algorithm receives at step $t$. Then choosing the loss function $\ell_t(\Tr[E_t \sigma_t]) = \vert \Tr[E_t (\sigma_t - \rho)]\vert$ results in $L=2$ and the total regret being exactly $\sum_{t=1}^T \vert \Tr[E_t (\sigma_t - \rho)]\vert$. We now show that the definition of $\QSSD$ allows $\A$ to find a series of POVM elements $\{E_t\}_{t=1}^T$ such that $\sigma_t$ is updated in a manner sufficient to solve the search problem.

    Let $\sigma_0 = \frac{1}{2^n}\id$ and $d = \QSSD_\tau^\beta$. By definition of $\QSSD$, there exists a measure $\mu_0$ on $\mathcal{F}$ such that $\QSDA(\Cc \backslash \Z_\delta(\mu_0), \sigma_0) \leq d$. By definition of $\QSDA$, there then exists a distribution $\eta$ over observables such that 
    \begin{align}
        \Pr_{M \sim \eta}\left[\vert \Tr[M(\rho' - \sigma_0)]\vert > \tau \right] \geq \frac{1}{d} 
    \end{align}
    for all $\rho' \in \Cc \backslash \Z_\delta(\mu_0)$. Assume now that $\rho \in \Cc \backslash \Z_\delta(\mu_0)$. Draw $d\cdot \ln(T/\delta)$ samples from $\eta$, denoted by $\{M_i\}_i$. Query each $M_i$ with tolerance $\tau / 3$ and let $v_i$ be the response. Let $i$ be such that $\vert \Tr[M_i(\rho-v_i)]\vert > 2\tau/3$. This will hold if $\Tr[M_i(\rho-\sigma_0) > \tau]$, which fails to occurs with probability at most $(1-\frac{1}{d})^{d\cdot \ln(T/\beta)} \leq \beta/T$. Define $E_0 = \frac{M_i + \id}{2}$. Then we have
    \begin{align}
        \vert \Tr[E_0 (\rho-\sigma_0)]\vert & = \frac{1}{2}\vert \Tr[M_i (\rho -\sigma_0)]\vert = \frac{1}{2}\vert (\Tr[M_i\rho]-v_i) +(v_i- \Tr[M_i\sigma_0)] > \frac{\tau}{6}\ .
    \end{align}
    Now update $\sigma_0$ using the online learning algorithm and POVM element $E_0$. Continue this process by using the definition of $\QSSD$ to find some distribution over observables and solutions, sample that distribution, and update $\sigma_t$ with $E_t$. The process stops when an $M_i$ distinguishing $\rho$ from $\sigma_t$ is not drawn, which occurs with some small probability of error or if $\rho \in \Z_\delta(\mu_t)$.

    In each step the online learning algorithm incurs a loss of at least $\tau / 6$, and thus a total regret of at least $T\cdot \tau / 6$. However, since the total regret is upper bounded by $4\sqrt{(2\log 2) Tn}$, this implies that $T \leq (1152 \ln 2)n/{\tau^2}$.     When $\A$ outputs a solution, if $\rho \in \Z_\alpha(\mu_T)$ then $\A$ succeeds with probability at least $\alpha$. If $\rho \notin \Z_\alpha(\mu_T)$ then the algorithm must have missed an update at some step. This occurs with probability at most $\beta$. Thus, the algorithm succeeds with probability at least $\alpha-\beta = 1-\delta$.
\end{proof}
It is worth noting that this result holds in an information-theoretic sense. The operators used to achieved such a query complexity may not efficiently constructed. Further, the theorem is not constructive. Nevertheless, this shows that $\QSSD$ \emph{characterizes} the complexity of search~problems.

\section{Alternative Proofs using Variance and Yao's Principle}\label{app:yao}

Our main separation between $\QSQ$ and $\QPAC$ involved using the variance method to lower bound $\QSDA$, which in turn lower bounds the learning problem. However, again by upper bounding the variance we can directly bound the learning complexity rather than going through an intermediate step involving decision problems. While this does not improve the lower bound on query complexity, it improves the constant below which learning is hard. Further, in some sense it is a simpler/more direct proof. In particular, we avoid having to make the assumption that $2\tau < \min_{f\in \Cc} \TRD(\psi_f, \E_f[\psi_f]) - 2\eps$. We give a proof here for reference and then discuss how this technique can be used in general.

\begin{theorem}\label{thm:mainsepalt}
 The concept class
    $$
    \Cc=\Big\{\ket{\psi_A}=\frac{1}{\sqrt{2^n}}\sum_{x\in \01^n} \ket{x,x^\top Ax \text{ (mod 2} )}:A\in \mathbb{F}_2^{n\times n}\Big\}
    $$ 
    requires $2^{\Omega(n)}$ many $\Qstat(1/\poly(n))$ queries to learn  below error $\sqrt{7}/8$ in trace distance.
\end{theorem}
\begin{proof}
    In the main text we saw that for any observable $M$ such that $\Vert M \Vert \leq 1$ we have that $\Var_A(\Tr[M\psi_A]) = O(2^{-n/2})$. Say that a deterministic algorithm $\A$ makes $t$ queries $\{M_i\}_{i=1}^t$ of tolerance $\tau$. Then by Chebyshev's inequality and a union bound we have that
    \begin{align}
        \Pr_{A\sim \Cc}[\exists i: \ \vert \Tr[M_i (\psi_A - \frac{1}{2^n}\id)] \vert > \tau] = O(\frac{t}{\tau^2 \cdot 2^{n/2}})\ .
    \end{align}
    Since the distribution above is uniform over $\Cc$, this implies that there is a subset of concepts $\Cc_0 \subseteq \Cc$ of measure at least $1-O(\frac{t}{\tau^2 \cdot 2^{n/2}})$ all consistent with the answer $\frac{1}{2^n}\Tr[M_i]$ to each query $M_i$. Then upon input of any $\psi_A \in \Cc_0$, $\A$ outputs the same solution $\pi$. We assume that $\A$ succeeds to output some $\pi'$ such that $\TRD(\psi_A, \pi') \leq $ on concepts not $A \notin \Cc_0$. However, as we will now show, $\A$ must fail to do so for almost all $A \in \Cc_0$.     Recall that for pure states we have that $\TRD(\psi_A, \psi_B) = \sqrt{1-\vert \braket{\psi_A}{\psi_B}\vert^2}$ and that for function states $\braket{\psi_f}{\psi_h} = \Pr_{x\sim\01^n}[f(x) = h(x)]$. Using Fact~\ref{fact:zippel} we have that $\Pr_{x\sim\01^n}[f_A(x) = f_B(x)] \leq \frac{3}{4}$ and thus that $\TRD(\psi_A, \psi_B) \geq \sqrt{7}/4$ if $A \neq B$. Now assume that there is some $A \in \Cc_0$ such that $\TRD(\psi_A, \pi) < \sqrt{7}/8$. Then for all $B \in \Cc_0$ such that $A \neq B$ by triangle inequality we have that
    \begin{align}
        \TRD(\psi_B,\pi) & \geq \TRD(\psi_A,\psi_B) - \TRD(\psi_A,\pi) > \sqrt{7}/8\ .
    \end{align}
    Thus, $\A$ fails on all other concepts in $\Cc_0$. Since the measure of $\Cc_0$ is at least $1-O(\frac{t}{\tau^2 \cdot 2^{n/2}})$, if $t = o(\tau^2 \cdot 2^{n/2})$ there is some $n_\delta$ such that for all $n > n_0$ we have that the measure of $\Cc_0$ is at least $\delta$ for any constant $0 < \delta < 1$. That is, for any $0 < \delta < 1$ we must have that $\A$ fails on at least a $\delta-$fraction of the concepts for $n$ large enough. By Yao's Principle, any randomized algorithm that succeeds with probability at least $1-\delta/2$ must then use $\Omega(\tau^2 \cdot 2^{n/2})$ $\Qstat$ queries of tolerance $\tau$.
\end{proof}
Note that this then in turn improves theorem~\ref{thm:smallckt} and theorem~\ref{thm:smallckt_distrib} to hold for learning the states below trace distance $\sqrt{7}/8$ and total variational distance $7/128$ respectively.
It is not hard to see how this technique generalizes to other concept classes. Say that $\Cc$ is some concept class such that $\Var_{\rho \sim \Cc}(\Tr[M\rho]) \leq v$. Then for $t = o(\tau^2 \cdot v)$ queries there are adversarial responses such that any deterministic algorithm $\A$ must output the same answer for a set of measure $\delta$ for any constant $0 < \delta < 1$. Denote this set by $\Cc_0$. Assume further that for any fixed state $\pi$ we have that $\Pr_{\rho\sim \Cc_0}[\TRD(\psi_f, \pi) > \eps] \geq \delta'$. Then we have that $\A$ fails to output a state below trace distance $\eps$ for a set of measure at least $\delta'\cdot \delta$. By Yao's principle, then any randomized algorithm must use $\Omega(\tau^2 \cdot v)$ queries of tolerance $\tau$ to output a $\pi$ such that $\TRD(\psi_f, \pi) < \eps$ with probability at least $1-\delta' \cdot \delta/2$. When it is easy to prove that any output must fail on any large subset of concepts this method readily gives direct lower bounds on $\QSQ(\Cc)$. When this property is hard to verify, going through $\QSD$ may be easier in practice.
\end{document}